\newcommand{\FullOrShort}{full}   
\newcommand{\final}{nonfinal}
\newcommand{\CommentsOnOffError}{error}  

\RequirePackage{ifthen}

\ifthenelse{\equal{\FullOrShort}{full}}{
		\documentclass[11pt]{article}
	  \newcommand{\fullOnly}[1]{#1}
	  \newcommand{\shortOnly}[1]{}
  }{
		\documentclass[conference]{IEEEtran}
	  \newcommand{\fullOnly}[1]{}
	  \newcommand{\shortOnly}[1]{#1}
  }

\ifthenelse{\equal{\CommentsOnOffError}{on}}{
\newcommand{\commentbh}[1]{{\bfseries BH: #1}}
\newcommand{\commentmm}[1]{{\bfseries MM: #1}}
}{}
\ifthenelse{\equal{\CommentsOnOffError}{off}}{
\newcommand{\commentbh}[1]{}
\newcommand{\commentmm}[1]{}
}{}

\usepackage[cmex10]{amsmath}
\usepackage{amssymb,amsthm}
\usepackage{cite}
\usepackage{color}
\usepackage{graphicx}
\usepackage{algorithm}
\usepackage{algorithmicx}
\usepackage[noend]{algpseudocode}
\usepackage[capitalize]{cleveref}
\usepackage{xspace}

\newcommand{\refalgSimpleCompute}{Algorithm 1\xspace}
\newcommand{\refalgRobustRandomnessExchange}{Algorithm 2\xspace}
\newcommand{\refalgComputeOblivious}{Algorithm 3\xspace}
\newcommand{\refalgComputeAdv}{Algorithm 4\xspace}

\ifthenelse{\equal{\final}{final}}{
		\newcommand{\finalOnly}[1]{#1}
	  \newcommand{\finalDel}[1]{}
		\renewcommand{\paragraph}[1]{\smallskip \emph{#1: }}
		\newcommand{\algorithmfontsize}{\footnotesize}
		\renewcommand{\Statex}{}
  }{
	  \newcommand{\finalOnly}[1]{}
	  \newcommand{\finalDel}[1]{#1}
		\newcommand{\algorithmfontsize}{\small}
		\usepackage{fullpage}
  }

\newcommand{\Trans}{\texttt{T}}
\newcommand{\MPone}{\texttt{MP1}}
\newcommand{\MPtwo}{\texttt{MP2}}
\newcommand{\VP}{\tilde{\texttt{k}}}
\newcommand{\verif}{\texttt{k}}
\newcommand{\voteone}{\texttt{v1}}
\newcommand{\votetwo}{\texttt{v2}}
\newcommand{\error}{\texttt{E}}
\newcommand{\BVC}{BVC}

\newcommand{\stringx}{\texttt{X}}
\newcommand{\stringy}{\texttt{Y}}
\newcommand{\strings}{\texttt{S}}

\newcommand{\twopartdef}[4]
{
	\left\{
		\begin{array}{ll}
			#1 & \mbox{if } #2  \smallskip \\
			#3 & \mbox{if } #4
		\end{array}
	\right.
}

\newcommand{\floor}[1]{\left\lfloor #1 \right\rfloor}
\newcommand{\ceil}[1]{\left\lceil #1 \right\rceil}

\newtheorem{theorem}{Theorem}[section]
\newtheorem{template}{Template}[section]

\newtheorem{definition}[theorem]{Definition} 
\newtheorem{lemma}[theorem]{Lemma}
\newtheorem{conjecture}[theorem]{Conjecture}
\newtheorem{corollary}[theorem]{Corollary}

\Crefname{theorem}{Theorem}{Theorems}
\Crefname{lemma}{Lemma}{Lemmas}
\Crefname{claim}{Claim}{Claims}
\Crefname{remark}{Remark}{Remarks}
\Crefname{corollary}{Corollary}{Corollaries}
\Crefname{proposition}{Proposition}{Propositions}
\Crefname{definition}{Definition}{Definitions}
\Crefname{observation}{Observation}{Observations}

\begin{document}


\date{}
\title{Interactive Channel Capacity Revisited}


\author{Bernhard Haeupler\\ Microsoft Research\\ \texttt{haeupler@cs.cmu.edu} }

\maketitle

\shortOnly{\setcounter{page}{0}}
\thispagestyle{empty}

\newcommand{\listen}{\text{listen}}
\newcommand{\eps}{\epsilon}

\begin{abstract}
We provide the first capacity approaching coding schemes that robustly simulate any interactive protocol over an \emph{adversarial channel} that corrupts any $\eps$ fraction of the transmitted symbols.  Our coding schemes achieve a communication rate of $1 - O(\sqrt{\eps \log \log 1/\eps})$ over any adversarial channel. This can be improved to $1 - O(\sqrt{\eps})$ for random, oblivious, and computationally bounded channels, or if parties have shared randomness unknown to the channel. 

\finalDel{\medskip}

\newcommand{\firstfootnote}{Our protocols work for the standard setting in which the input protocol is alternating and the simulation has an alternating or non-adaptive, i.e., fixed, communication order. The impossibility result of \cite{KR} does not hold for alternating input protocols. Instead, an input protocol with a more complex communication order is assumed while the simulations are restricted to be non-adaptive. We point out that insisting on non-adaptive simulations is too restrictive for  general input protocols: Independently of the amount of noise most (non-alternating) input protocols \emph{cannot} be simulated robustly in a non-adaptive manner, i.e., a rate of $1 - o(1)$ is impossible even if the channel introduces merely a single random error. The $1 - O(\sqrt{H(\eps)})$-rate coding scheme of \cite{KR} avoids this barrier by restricting the input protocols that can be simulated. Our coding scheme naturally works for \emph{any} input protocol by allowing adaptive coding schemes as introduced in \cite{GHS14}.}

Surprisingly, these rates exceed\finalDel{\footnote{\firstfootnote}} the $1 - \Omega(\sqrt{H(\eps)}) = 1 - \Omega(\sqrt{\eps \log 1/\eps})$ interactive channel capacity bound which [Kol and Raz; STOC'13] recently proved for \emph{random errors}. We conjecture $1 - \Theta(\sqrt{\eps \log \log 1/\eps})$ and $1 - \Theta(\sqrt{\eps})$ to be the \emph{optimal} rates for their respective settings and therefore to capture the \emph{interactive channel capacity} for \emph{random and adversarial errors}.

\finalDel{\medskip}

In addition to being very communication efficient, our randomized coding schemes have multiple other advantages. They are computationally efficient, extremely natural, and significantly simpler than prior (non-capacity approaching) schemes. In particular, our protocols do not employ any coding but allow the original protocol to be performed \emph{as-is}, interspersed only by short exchanges of hash values. When hash values do not match, the parties backtrack. 
Our approach is, as we feel, by far the simplest and most natural explanation for why and how robust interactive communication in a noisy environment is possible. 

\end{abstract}

\finalDel{\newpage}

\section{Introduction}

We study the \emph{interactive channel capacity} of random and adversarial error channels, that is, the fundamental limit on the communication rate up to which any interactive communication can be performed in the presence of noise. We give novel coding schemes which we conjecture to achieve the corresponding interactive channel capacity for a wide variety of channels, up to a constant in the second order term for any small error rate $\eps$. Our coding schemes are extremely simple, computationally efficient, and give the most natural and intuitive explanation for why and how error correction can be performed in interactive communications. 

\subsection{Prior Work}

\paragraph{From Error Correcting Codes $\ldots$}
The concept of (forward) error correcting codes has fundamentally transformed the way information is communicated and stored and has had profound and deep connections in many sub-fields of mathematics, engineering, and beyond. Error correcting codes allow to add redundancy to any message consisting of $n$ symbols, e.g., bits, and transform it into a coded message with $\alpha n + o(n)$ symbols from a finite alphabet $\Sigma$ from which one can recover the original message even if any $\eps$-fraction of the symbols are corrupted in an arbitrary way. This can be used to store and recover information in a fault tolerant way (e.g., in CDs, RAM, $\ldots$) and also leads to robust transmissions over any noisy channel. In the later case, one denotes with $R = \frac{1}{\alpha \log |\Sigma|}$ the \emph{communication rate} at which such a communication can be performed with negligible probability of failure and for a given channel one denotes with the \emph{channel capacity} $C$ the supremum of the achievable rates for large $n$.

The groundbreaking works of Shannon and Hamming showed that the capacity $C$ of any binary channel with an $\eps$ fraction of noise satisfies $C = 1 - \Theta(H(\eps))$, where $H(\eps) = \eps \log \frac{1}{\eps} + (1 - \eps) \log \frac{1}{1 - \eps}$ denotes the binary entropy function which behaves like $H(\eps) = \eps \log \frac{1}{\eps} + O(\eps)$ for $\eps \rightarrow 0$. More precisely, for random errors as modeled by the \emph{binary symmetric channel} (BSC), which flips each transmitted bit with probability $\eps$, Shannon's celebrated theorem states that the rate $R = 1 - H(\eps)$ is the exact asymptotic upper and lower bound for achieving reliable communication. Furthermore, for arbitrarily, i.e., adversarially, distributed errors Hamming's work shows that the rate $R = 1 - \Theta(H(\eps))$ remains achievable. Determining the optimal rate of a binary code or even just the constant hidden in the second order term is a fundamental and widely open question in coding theory. The popular guess or conjecture is that the Gilbert-Varshamov bound of $R \leq 1 - H(2 \eps)$ is essentially tight.


\paragraph{$\ldots$ to Coding Schemes for Interactive Communications}
These results apply to \emph{one-way} communications in which one party, say Alice, wants to communicate information to another party, say Bob, in the presence of 
 noise. In this paper we are interested in the same concept but for settings in which Alice and Bob have a \emph{two-way} or \emph{interactive} communication which they want to make robust to noise by adding redundancy. More precisely, Alice and Bob have some conversation in mind which in a noise-free environment can be performed by exchanging $n$ symbols in total. 
They want a \emph{coding scheme} which adds redundancy and transforms any such conversation into a slightly longer $\alpha n$-symbol conversation from which both parties can recover the original conversation outcome even when any $\eps$ fraction of the coded conversation is corrupted.

The reason why one cannot simply use error correcting codes for each message is that misunderstanding just one message, which corresponds only to a $1/n$ fraction of corruptions for an $n$ message interaction, leads to the remainder of the conversation becoming irrelevant and useless. It is therefore a priori not clear that tolerating some (even tiny) constant fraction of errors $\eps$ is even possible. 

In 1993 Schulman\cite{Schulman} was the first to address this question. In his beautiful and at the time surprising work he showed that tolerating an $\eps = 1/240$ fraction of the adversarial errors is possible for some constant overhead $\alpha = \Theta(1)$. This also directly implies that any error rate bounded away by a constant from $1/2$ can be tolerated for the easier random errors setting, since one can easily reduce the error rate by repeating symbols multiple times. Later, Braverman and Rao~\cite{BR11} showed that an adversarial error rate of up to $\eps < 1/4$ could be tolerated with a constant overhead $\alpha = \Theta(1)$. Lastly, \cite{GHS14,GH13,EffremenkoBraverman,FGOS13,EGeH14} determined the full error rate region in which a non-zero communication rate is possible in a variety of settings. While the initial coding schemes were not computationally efficient, because they relied on powerful but hard to construct tree codes, later results \cite{BK12,BN13,GH13,GMS11} provided polynomial time schemes using randomization. 

\paragraph{Capacity Approaching Coding Schemes}
None of these works considers the communication rate that can be maintained for a given noise level $\eps$. In fact, each of these schemes has a relatively large unspecified constant factor overhead $\alpha$ even for negligible amounts of noise $\eps \rightarrow 0$. This is unsatisfactory as one would expect the necessary amount of redundancy to vanish in this case. However, capacity approaching schemes were considered out of scope of techniques up to this point~\cite[Problem 8]{BravermanAllerton}. The work by Kol and Raz~\cite{KR} was the first to consider the communication rate up to which interactive communication could be maintained over a noisy channel. In particular, they studied random errors, as modeled by the BSC, with an error probability $\eps \rightarrow 0$. Under seemingly reasonable assumptions they proved an upper bound of $1 - \Omega(\sqrt{H(\eps)})$, that is, that some protocols cannot be robustly simulated with a rate exceeding $1 - \Omega(\sqrt{H(\eps)})$. This is significantly lower than the $1 - H(\eps)$ bound for the standard one-way setting. They also gave a coding scheme that achieves a matching rate of $1 - O(\sqrt{H(\eps)})$ for the same setting. This was seen as a characterizing the BSC \emph{interactive channel capacity} up to constants in the second order term, a notable breakthrough.

\subsection{Our Results}\label{sec:ourresults}

\paragraph{Coding Schemes for Adversarial Errors}
This work started out as an attempt to address the question of communication rate and interactive channel capacity in the much harsher adversarial noise setting. In particular, the hope was to design a \emph{capacity approaching} coding scheme for adversarial errors that could achieve a communication rate approaching one as the noise level $\eps$ goes to zero.

To our shock the communication rates of our final coding schemes exceed the $1 - \Omega(\sqrt{H(\eps)}) = 1 - \Omega(\sqrt{\eps \log \frac{1}{\eps}})$ bound of \cite{KR} even though they work in the much harder adversarial noise setting. In particular, the new communication schemes operate at a communication rate of $1 - O(\sqrt{\eps})$ against any worst-case oblivious channel and even channels controlled by a fully adaptive adversary as long as this adversary is computationally efficient or as long as the parties share some randomness unknown to the adversary. These trivially include the case of i.i.d.\ random errors. For the unrestricted fully adaptive adversarial channel we achieve a rate of $1 - O(\sqrt{\eps \log\log \frac{1}{\eps}})$ which is still lower than the bound of the impossibility result of \cite{KR}.

\paragraph{Interactive Channel Capacity Revisited}
We uncover that these differences stem from important but subtle variations in the assumptions on the \emph{communication order}, that is, the order in which which Alice and Bob speak, both for the original (noiseless) input protocol $\Pi$ and for its \finalDel{noise-robust} simulation $\Pi'$:

The standard setting, used by all works prior to \cite{KR}, assumes that $\Pi$ is alternating, that is, that the parties take turns sending one symbol each. Both our $1 - O(\sqrt{\eps})$-rate and our $1 - O(\sqrt{\eps \log \log \frac{1}{\eps}})$-rate coding scheme as well as the $1 - O(\sqrt{H(\eps)})$-rate coding scheme of \cite{KR} work in this setting. For all three protocols the simulation $\Pi'$ can furthermore be chosen to have the same fixed alternating communication order. We remark that if one does not care about constant factors in the communication rate assuming an alternating input protocol is essentially without loss of generality because any protocol can be transformed to be alternating while only increasing its length by at most a factor of two. However, this transformation and wlog-assumption cannot be applied to design capacity approaching protocols. Regardless, the setting of alternating protocols and simulations is a very simple and clean setting which still encapsulates the characteristics and challenges of the problem.

The impossibility result of \cite{KR} however does not apply to alternating protocols. Instead, an input protocol with a more complex communication order is assumed. More importantly, the simulations $\Pi'$ are restricted to be non-adaptive, that is, have an a priori fixed communication order which defines for each time step which party sends and which listens. The $1 - \Omega(\sqrt{\eps \log 1/\eps})$ lower bound of \cite{KR} subtly but nonetheless crucially builds on this non-adaptivity assumption. 

After understanding these issues better, we point out that insisting on non-adaptive simulations $\Pi'$ is too restrictive for general input protocols $\Pi$ in a very strong sense: Many non-alternating input protocols $\Pi$ simply cannot be simulated robustly without losing at least a constant factor in the communication rate. This impossibility is furthermore essentially unrelated to the type of noise in the channel and therefore unrelated to the question of interactive channel capacity. More precisely, we conjecture the following $1 - \Omega(1)$ impossibility result to hold\finalDel{ (see also \Cref{sec:importantremarks})}:

\begin{conjecture}\label{claim:adaptiveunstructured}
Any protocol $\Pi$ with a sufficiently non-regular, e.g., pseudo-random, communication order cannot be robustly simulated by any non-adaptive protocol $\Pi'$ with a rate of $R = 1 - o(1)$. This is true for essentially any channel introducing some error, in particular, even for channels introducing merely a single random error or erasure.
\end{conjecture}

In this work we show that there is a natural way to circumvent this impossibility barrier without having to restrict the input protocols $\Pi$ that can be simulated. In particular, our protocols very naturally simulate \emph{any} general input protocol $\Pi$ if the simulation $\Pi'$ is allowed to be have an adaptive communication order, as introduced in \cite{GHS14}.

We furthermore conjecture that the $1 - O(\sqrt{\eps})$ bound we present in this work is the natural and tight bound on the maximum rate that can be achieved in a wide variety of interactive communication settings:

\begin{conjecture}\label{claim:capacity}
The maximal rate achievable by an interactive coding scheme for any binary error channel with random or oblivious errors is $1 - \Theta(\sqrt{\eps})$ for a noise rate $\eps \rightarrow 0$. This also holds for for fully adversarial binary error channels if the adversary is computationally bounded or if parties have access to shared randomness that is unknown to the channel. It also remains true for all these settings regardless whether one restricts the input protocols to be alternating or not. 
\end{conjecture}

We suspect this claim also extends to larger alphabets. We feel that the broadness and robustness of this bound might justify regarding $1 - \Theta(\sqrt{\eps})$ as \emph{the} interactive channel capacity of a random error channel, even though this paper clearly demonstrates how careful and precise one needs to specify how protocols can use a channel before being able to talk about rates and capacities. 

For the fully adversarial with no shared randomness and a binary channel alphabet we conjecture our rate of $1 - O(\sqrt{\eps \log \log \frac{1}{\eps}})$ to be tight as well. This bound does not hold for larger alphabets (see \finalDel{\Cref{sec:importantremarks}}\finalOnly{\cite{H14arxiv}}):

\begin{conjecture}\label{claim:capacityadv}
The interactive channel capacity for the fully adversarial binary error channels in the absence of shared randomness is $1 - \Theta\left(\sqrt{\eps \log \log \frac{1}{\eps}}\right)$ for a noise rate $\eps \rightarrow 0$. 
\end{conjecture}

Lastly, we remark that subsequent to this work it was shown in \cite{GeH14} that a higher rate of $1 - \Theta(H(\eps))$ is possible for coding schemes that robustly simulate any (alternating) protocol over random or adversarial channels with feedback or over random or adversarial erasure channels.

\paragraph{Simple, Natural, Communication and Computationally Efficient Coding Schemes}
In addition to being capacity approaching for worst case errors with an optimal asymptotic dependence on the noise level $\eps$ the new coding schemes also have the advantage of being much simpler and more natural than prior coding schemes tolerating adversarial errors. In particular, they are essentially the first such schemes that do not rely on tree-codes. In fact, they do not perform any coding at all. Instead, they operate along the following extremely natural template, which so far seemed only applicable to the setting with random noise~\cite{Schulman92,KR}:

\begin{template}[Making a Conversation robust to Noise] \mbox{}\\
Both parties have their original conversation as if there were no noise except that
\begin{enumerate}
	\item sporadically a concise summary (an $\Theta(1)$ or $\Theta(\log \log \frac{1}{\eps})$ bit random hash value) of the conversation up to this point is exchanged.
	\item If the summaries match the conversation continues. 
	\item If the summaries do not match, because the noise caused a misunderstanding, then the parties backtrack.
\end{enumerate}
\end{template}

Some details go into how to compute the summaries and how to coordinate the backtracking steps. Still, the protocol stays simple and this outline is so intuitive that it can be easily explained to non-experts. In that respect it can be seen as demystifying Schulman's result that interactive communication can be performed at a constant rate in the presence of a constant fraction of adversarial noise. Our proofs for why these constant or $\Theta(\log \log \frac{1}{\eps})$ size hash values are sufficient (and necessary) are simple and completely elementary with the standard hash functions from \cite{NaorNaor} being the only non-trivial black-box used. Furthermore, since our alternative proof is based solely on hashing it directly leads to a computationally efficient ``coding'' scheme. This scheme works without any assumptions on the structure of the original protocol and is so simple that real-world use-cases and implementations become a possibility. In fact, Microsoft has a utility patent pending. All in all, we feel that this paper gives the simplest, most natural, and most intuitive explanation for why robust interactive communication in an environment with adversarial noise is possible and how it can be achieved.

\subsection{Organization}	

The remainder of this paper is organized as follows. In \Cref{sec:models} we provide preliminaries and the channel and interactive communication models. In \Cref{sec:impossibility} we explain the fundamental difference between error correction for interactive communications and the classical one-way setting. In particular, we explain why a communication rate of $1 - \Omega(\sqrt{\eps})$ is the best one can hope for in an interactive setting. In \Cref{sec:simpleprot} we provide a simple coding scheme achieving this rate against adversarial errors if the channel operates on a  logarithmic bit-size alphabet. In \Cref{sec:problemsandsolutions} we explain the barriers in extending this algorithm to a constant size alphabet and give an overview of the techniques we use to overcome them. In \Cref{sec:protocol} we then provide our new coding schemes and we use \Cref{sec:proof} to prove their correctness. 

\section{Definitions and Preliminaries}\label{sec:models}


\finalDel{
In this section, we define the interactive coding setting including \emph{adaptive interactive protocols} as defined by \cite{GHS14}. We provide several important remarks regarding the use of this setting in this paper \finalOnly{in the full version \cite{H14arxiv}.} \finalDel{and in \cite{KR} in \Cref{sec:importantremarks}.}
}

%

\subsection{Interactive Protocols and Communication Order}

An \emph{interactive protocol} $\Pi$ defines some communication performed by two parties, Alice and Bob, over a channel with alphabet $\Sigma$. After both parties are given an input the protocol operates in $n$ \emph{rounds}. For each round of communication each party decides independently whether to listen or transmit a symbol from $\Sigma$ based on its state, its input, its randomness, and the history of communication, i.e., the symbols received by it so far. All our protocols will utilize \emph{private randomness} which is given to each party in form of its own infinite string of independent uniformly random bits. It is also interesting to consider settings with \emph{shared randomness} in which both parties in every round $i$ have access to the same infinite random bit-string $R_i$. 

We call the order in which Alice and Bob speak or listen the \emph{communication order} of a protocol. Prior works have often studied \emph{non-adaptive} protocols for which this communication order is predetermined. In this case, which player transmits or listens depends only on the round number and it is deterministically ensured that exactly one party transmits in each round. If the such a non-adaptive communication order repeats itself in regular intervals we call the protocol \emph{periodic} and call the smallest length of such an interval the \emph{period}. The simplest communication order has Alice and Bob take taking turns. We call such a protocol, with period two, \emph{alternating}.

\subsection{Adversarial and Random Communication Channels}

The communication between the two parties goes over a \emph{channel} which delivers a possibly corrupted version of the chosen symbol of a transmitting party to a listening party. In particular, if exactly one party listens and one transmits then the listening party receives the symbol chosen by the transmitting party unless the channel interferes and corrupts the symbol.

In the \emph{fully adversarial channel} model with \emph{error rate} $\eps$ the number of such interferences is at most $\eps N$ for an $N$ round protocol and the adversary chooses the received symbol arbitrarily. In particular, the adversary gets to know the length $N$ of the protocol and therefore also how many corruptions it is allowed to introduce. In each round it can then decide whether to interfere or not and what to corrupt a transmission to based on its state, its own randomness, and the communication history observed by it, that is, all symbols sent or received by the parties so far. The adversary does not get to know or base its decision on the private randomness of Alice or Bob (except for what it can learn about it through the communicated symbols). In the shared randomness setting we differentiate between the default setting, in which the adversary gets to know the shared randomness as well, that is, base its decisions in round $i$ also on any $R_j$ with $j \leq i$, and the \emph{hidden shared randomness} setting in which the the shared randomness is a secret between Alice and Bob which the adversary does not know. 

We also consider various relaxations of this all powerful fully adversarial channel model: We call an adversary \emph{computationally bounded} if the decisions of the adversary in each round are required to be computable in time polynomial in $N$. We call an adversary \emph{oblivious} if it makes all its decisions in advance, in particular, independently of the communication history. A particular simple such oblivious adversary is the \emph{random error channel} which introduces a corruption in each round independently with probability $\eps$. We will consider random channels mostly for binary channels for which a corruption is simply a bit-flip.

\subsection{Adaptive Interactive Protocols}

It is natural and in many cases important to also allow for \emph{adaptive} protocols which base their communication order decisions on the communication history.\finalDel{ In particular, it is natural to design adaptive protocols in which both parties decide separately which part of the original protocol to simulate based on (estimates of) where errors have happened so far. When these error estimates differ this can lead to rounds in which both parties transmit or listen simultaneously, in particular when the protocol to be simulated is non-periodic.} We follow \cite{GHS14} in formalizing the working of the channel in these situations. In particular, in the case of both parties transmitting no symbol is delivered to either party because neither party listens anyway. In the case of both parties listening the symbols received are undetermined with the requirement that the protocol works for any received symbols. In many cases it is easiest to think of the adversary being allowed to choose the received symbols, without it being counted as a corruption.

\subsection{Robust Simulations}

A protocol $\Pi'$ is said to \emph{robustly simulate} a deterministic protocol $\Pi$ over a channel $C$ if the following holds: Given any inputs to $\Pi$, both parties can (uniquely) decode the transcript of the execution of $\Pi$ over the noise free channel on these inputs from the transcript of an execution of $\Pi'$ over the channel $C$. For \emph{randomized protocols} we say protocol $\Pi'$ \emph{robustly simulates} a protocol $\Pi$ with \emph{failure probability} $p$ over a channel $C$ if, for any input and any adversary behind $C$, the probability that both parties correctly decode is at least $1 - p$. We note that the simulation $\Pi'$ typically uses a larger number of rounds, e.g., $\alpha n$ rounds for some $\alpha > 1$. 

%
%
%
%

\newcommand{\importantremarks}{

\fullOnly{\subsection{Important Remarks Regarding the Interactive Coding Settings}\label{sec:importantremarks}}

\paragraph{Structure of the Communication Order of the Original Protocol} 
It is possible to convert the original protocol, which is to be simulated, into an alternating protocol while only increasing the amount of communication by a factor of two. This was used in all prior works essentially without loss of generality because it preserves the overhead $\alpha$ and communication rate $R$ of the simulation up to a constant factor. However, in works concerned with achieving an overhead factor of $\alpha \rightarrow 1$, such as \cite{KR} and this work, any assumption on the original protocol having a very nice and regular structure is restricting the generality and applicability of results. In contrast to \cite{KR}, all results in this paper do not need to make any such assumptions.


\paragraph{Alphabet Sizes} 
Similarly, prior works which happily tolerated constant factor losses in the communication rate, assumed that the original protocol is binary while the simulation protocol uses a large finite alphabet. However, if one wants to determine communication rates or talk about the capacity of a channel it is natural, convenient, and essentially necessary to keep the alphabet of a simulation $\Pi'$ to be the same alphabet $\Sigma$ as in the original protocol $\Pi$. Throughout this paper, unless otherwise noted, we choose $\Sigma$ to be the binary alphabet $\Sigma = \{0,1\}$. This constitutes the hardest case and, in fact, all our algorithmic results work directly as described for any alphabet size. Furthermore, the rate for the adversarial channel can be easily improved/generalized to be $1 - O\left(\sqrt{\eps \max\left\{1,\frac{\log \log \frac{1}{\eps}}{\log |\Sigma|}\right\}}\right)$ which becomes $1 - \Theta(\sqrt{\eps})$ for alphabets of bit-size $\Theta(\log \log \frac{1}{\eps})$. We believe this bound to be the tight channel capacity bound for any alphabet size in the fully adversarial setting.

\paragraph{Adaptive Robust Simulations vs. Non-Adaptive Simulations with Predetermined Communication Order} 
In \cite{GHS14} extensive arguments are given for why the adaptive protocols we use here is the right ``adaptivity model'': Most importantly, it does not allow for any signaling or time-coding (e.g., transmitting bits by being silent/sending a message or encoding information in the length of silence between breaks) which could be incorruptible or allow to send more than $\log_2 |\Sigma|$ bits of information per symbol sent. Another nice property is that for non-adaptive protocols that perfectly coordinate a designated sender and receiver in each round our model matches the standard setting. In \cite{BR11} it was furthermore shown that any protocol that is not non-adaptive can lead to rounds in which the adaptive parties fail to coordinate a designated sender/receiver. In this case our model carefully precludes any information exchange. This matches the intuition that one should not be able to gain an advantage from both parties speaking or listening at the same time. This also provides the guarantee that the total amount of information exchanged between the two parties (in either direction) is at most $\log_2 |\Sigma|$. This is particularly important when considering communication rates and capacities as one can use $n \log_2 |\Sigma|$ as the baseline of how much information can at most be exchanged in $n$ rounds.

\paragraph{Why Adaptivity is Necessary in Simulating General Protocols} 
The importance of being adaptive when simulating non-regular protocols was already mentioned in \Cref{claim:adaptiveunstructured}. The reason for it is simple: During a simulation it is essentially unavoidable that, due to errors, the parties are temporarily in disagreement regarding which part of the original conversation is to be computed next. In particular, if the communication order of the original does not have a regular structure both parties might think they should speak next. A different way to see the same thing is to note that any simulation $\Pi'$ of a protocol $\Pi$ cannot be communication efficient if it does not have (nearly) the same communication order. However, if the communication order of $\Pi$ is non-regular and the communication order of $\Pi'$ needs to be chosen in advance then one needs to know a priori which part of $\Pi$ is to be computed at what time period of the execution of $\Pi'$. Since $\Pi$ needs to be simulated in order and the time to compute a block of $\Pi$ highly depends on the occurrences of errors during the simulation this is not possible. 

Taking these considerations into account it is clear that most non-regular protocols will be hard to simulate robustly and non-adaptively without losing a constant factor in the communication rate as stated in \Cref{claim:adaptiveunstructured}. The protocols for which this is most easily seen are protocols in which the communication order of the original protocol is pseudo random. For such a protocol two different parts of the protocol would not have communication orders that do not match in a constant fraction of the rounds. This means that if one tries to simulate a (full entropy) protocol with a such a pseudrandom communication order and the simulation gets off-synch by just one round, due to a deletion or erasure, the simulation will be progressing with a constant factor slower speed until extra rounds are introduced to catch up. If the position of the error is random and therefore not known in advance it is necessary to introduce ``extra steps'' every constant number of rounds on average which leads to the $1 - \Omega(1)$ bound stated in \Cref{claim:adaptiveunstructured}.
}
\fullOnly{\importantremarks}

\section{Channel Capacities for Interactive vs. One-way Communication}\label{sec:impossibility}

In this section we explain the important difference in correcting errors for interactive communications in contrast to the standard one-way setting of error correcting codes. We then quantify this difference and give the high-level argument for why $1 - \Omega(\sqrt{\eps})$ is the best possible communication rate one can expect for coding schemes that make interactive communications robust to even just random noise. The impossibility result of \cite{KR} can be seen as formalizing a very similar argument. 


\subsection{The Difficulties of Coding for Interactive Communications}

We first explain, on a very intuitive level, why making interactive communications resilient to noise is much harder than doing the same for one-way communications. In particular, we want to contrast the task of Alice transmitting some information to Bob with the task of Alice and Bob having a conversation, e.g., an interview of Bob by Alice, both in a noisy environment. The main difference between the two tasks is that in the one-way communication Alice knows everything she wants to transmit a priori. This allows her to mix this information together in an arbitrary way by using redundant transmissions that protects everything equally well. This contrasts with an interactive communication where Alice's transmissions depend highly on the answers given by Bob. In our interview example, Alice cannot really ask the second question before knowing the answer to her first question as what she wants to know from Bob might highly depend on Bob's first answer. Even worse, Alice misunderstanding the first answer might completely derail the interview into a direction not related to the original (noise free) conversation. In this case, everything talked about in the continuing conversation will be useless, even without any further noise or misunderstandings, until this first misunderstanding is detected. Lastly, in order to resolve a detected misunderstanding the conversation has to backtrack all the way to where the misunderstanding happened and continue from there.

\subsection{The $1 - \Omega(\sqrt{\eps})$ Fundamental Rate Limit}

Next, we aim to quantify this difference and argue for $1 - \Theta(\sqrt{\eps})$ being a fundamental rate limit of interactive communication, even for random errors. 
 We pick the simplest noise model and assume, for now, that Alice and Bob try to communicate over a binary symmetric channel, that is, a binary random error channel which flips each bit transmitted independently with the small error probability $\eps$. For the one-way communication task Alice can simply transmit everything she wants to send followed by a few check-sums over the complete message. It is a classical result that for a full error recovery it suffices to add to the transmission approximately $H(\eps)$ as many randomly picked  linear check-sums as transmitted symbols, which in the binary case are simply parities over a randomly chosen subset. This leads to a rate of $1 - H(\eps)$ which is also optimal. 

Now we consider Alice and Bob having an interactive conversation over the same channel. Because of the noise Alice and Bob will need to add some redundancy to their conversation at some point in order to at least detect whether a misunderstanding has happened. For this one (check-)bit is necessary\finalDel{\footnote{The fact that for a full-entropy protocol Alice and Bob cannot detect a failure without communicating $\Theta(1)$ symbols worth of entropy is non-trivial and crucial for this argument. We point out that this does not hold for erasure channels or channels with feedback as demonstrated in \cite{GeH14}. For the adversarial setting even more communication is necessary. In particular, in order to detect any two adversarial bit flips in an $\eps^{-\Theta(1)}$ long string $\log \log \frac{1}{\eps}$ (check-)bits are necessary. This quantity can be seen as the minimum seed length required to distinguish with constant probability between a string and all strings of hamming distance two. The proof for this bound follows the same idea as the proof for \Cref{lem:seedlengthLB}. This leads to the extra $\sqrt{\log \log \frac{1}{\eps}}$ factor in the capacity upper and lower bound for fully adversarial binary error channels.}}. Say they do this after $r$ steps. The likelihood for an error to have happened is $r \eps$ at this point and the length of the conversation that needs to be redone because of such a preceding misunderstanding is of expected length $\frac{r}{2}$. This leads to an expected rate loss of $\Theta(r \eps)$ because every $r$ steps an expected $\frac{r^2 \eps}{2}$ steps are wasted. With this reasoning one would like to make $r$ as small as possible. However, adding one unit of redundancy every $r$ steps leads to a rate loss of $\frac{1}{r}$ itself regardless of whether errors have occurred. Balancing $r$ to minimize these two different causes of rate loss leads to an optimal rate loss of $\Theta(\min_r \{r\eps + 1/r\}) = \Theta(\sqrt{\eps})$ assuming $r$ is set optimally to $r = \frac{1}{\sqrt{\eps}}$. This argument applies in essentially any interactive coding setting and explains why the fundamental channel capacity drops from $1 - H(\eps)$ to $1 - \Theta(\sqrt{\eps})$ for interactive communications.

\section{A Simple Coding Scheme for Large Alphabets}\label{sec:simpleprot}

In this section\finalOnly{, as a warmup,} we give a simple coding scheme that achieves a rate of $1 - \Theta(\sqrt{\eps})$ against any fully adversarial error channel, albeit while assuming that the original protocol and the channel operate on words, that is, on the same $\Theta(\log n)$ bit-sized alphabet. 

\finalDel{
There are at least two compelling reasons to start with describing this algorithm as a warmup: (1) The outline of the coding scheme and the structure of its proofs are closely related to those of our main results. While nicely demonstrating the type of analysis and its main components the proofs here are much simpler, shorter, and easier to follow. Also, given this simple coding scheme it is much easier to understand the problems and barriers that need to be addressed in the small alphabet setting. (2) We view the assumption that parties communicate via small logarithmic size messages, instead of on a bit by bit basis, as very reasonable\finalDel{\footnote{Super-constant alphabet sizes are (unfortunately?) not a very common assumption for the interactive coding setting. We note that, while they do make some details easier, such as, the construction of good tree codes~\cite{Schulman}, designing good coding schemes remains a highly non-trivial task even when one just wants to achieve some small constant communication rate against random errors.}}. In fact, such an assumption is standard in many areas, like the theory of distributed computing~\cite{peleg2000distributed}. 
 We believe this simple algorithm to be a great candidate for a practically relevant and easily implementable coding scheme. 
}


\finalOnly{

\subsection{Overview}

To design our coding scheme we follow the idea is outlined in \Cref{sec:ourresults} and turn the $1 - \Omega(\sqrt{\eps})$ bound from \Cref{sec:impossibility} it into a converse, that is, prove it to be tight by designing a protocol achieving this rate. Our protocol can also be seen as a drastically simplified version of \cite{BK12}. 

In particular, the parties start with performing the original interactive protocol without any coding for $r = \Theta(\frac{1}{\sqrt{\eps}})$ steps as if no noise is present at all. Both parties then try to verify whether an error has occurred. They do this by randomly sampling a hash function and sending both the description of the hash function as well as the hash value of their complete communication transcript up to this point. They use $r_c = O(1)$ symbols for this check and initiate a backtracking step if a mismatch is detected. 

The (standard) hash functions used for this and throughout this paper are derived from the $\eps$-biased probability spaces constructed in \cite{NaorNaor}. These hash functions give the following guarantees:

\begin{lemma}[\cite{NaorNaor}] \label{lem:hashes}
For any $n$, any alphabet $\Sigma$, and any probability $0<p<1$, there exist $s = \Theta(\log (n \log |\Sigma|) + \log \frac{1}{p})$, $o = \Theta(\log \frac{1}{p})$, and a simple function $h$, which given an $s$-bit uniformly random seed $S$ maps any string over $\Sigma$ of length at most $n$ into an $o$-bit output, such that the collision probability of any two $n$-symbol strings over $\Sigma$ is at most $p$. In short:
\finalOnly{$\forall n,\Sigma,0<p<1: \ \exists s = \Theta(\log (n \log |\Sigma|) + \log \frac{1}{p}), o = \Theta(\log \frac{1}{p}), h: \{0,1\}^s \times \Sigma^n \mapsto \{0,1\}^o\ s.t.$ $\forall \stringx,\stringy \in \Sigma^{\leq n}, \stringx \neq \stringy, \strings \in \{0,1\}^s \ iid\ Bernoulli(1/2): \ \ P[h_\strings(\stringx) = h_\strings(\stringy)] \leq p$}
\finalDel{%
$$\forall n,\Sigma,0<p<1: \ \exists s = \Theta(\log (n \log |\Sigma|) + \log \frac{1}{p}), o = \Theta(\log \frac{1}{p}), h: \{0,1\}^s \times \Sigma^n \mapsto \{0,1\}^o\ s.t.$$
$$\forall \stringx,\stringy \in \Sigma^{\leq n}, \stringx \neq \stringy, \strings \in \{0,1\}^s \ iid\ Bernoulli(1/2): \ \ P[h_\strings(\stringx) = h_\strings(\stringy)] \leq p$$
}
\end{lemma}

The setting of interest to our first algorithm is that any two $\Theta(n)$ bit strings can, with high probability, be distinguished by a random hash function which creates $o = \Theta(\log n)$ size fingerprints and requires a seed of only $\Theta(\log n)$ bits. This allows us to communicate both the selected hash function (seed) and the corresponding hash value of the $O(n)$ long transcript using only $r_c = \Theta(1)$ symbols of $\Theta(\log n)$ bits each.

A last minor technical detail is a simple preprocessing step in which the original protocol $\Pi$ is modified by adding $\Theta(\sqrt{\eps} n)$ steps at the end, in which both parties send each other a fixed symbol. These steps should be interpreted as confirmation steps which reaffirm the correctness of the previous conversation. This ensures that the simulation $\Pi'$ never runs out of steps of $\Pi$ to simulate.

}

\finalDel{	
\subsection{Overview}

\subsubsection{High Level Idea}

To design our coding scheme we follow the argument of the $1 - \Omega(\sqrt{\eps})$ impossibility result from \Cref{sec:impossibility} and convert it into a converse, that is, prove it to be tight by designing a protocol achieving this rate. 

The general idea is as outlined in \Cref{sec:ourresults} and can also be seen as a drastically simplified version of \cite{BK12}. In particular, the parties start with performing the original interactive protocol without any coding for $r$ steps as if no noise is present at all. Both parties then try to verify whether an error has occurred. They do this by randomly sampling a hash function and sending both the description of the hash function as well as the hash value of the their complete communication transcript up to this point. They use $r_c$ symbols for this check. They then evaluate the hash function they received from the other party on their own transcript and check whether the hash values match. If they do, a party simply continues the conversation for the next $r$ steps. If the hash values do not match for a party then this party backtracks. 

\subsubsection{Setting the parameters $r$ and $r_c$}\label{sec:settingrandrc}

We now look at how one needs to set the parameters $r$ and $r_c$ in order to arrive at a robust protocol with optimal communication rate. In essence, we want to choose $r$ and $r_c$ such that the communication overhead which is done in addition to the $n$ rounds of original communication is a small fraction in dependence on the error rate $\eps$. Similar to the impossibility result the communication overhead in our general approach comes from two places: 

Firstly, verification overhead is the fraction of communication which is dedicated to communicating verification information instead of performing the original protocol. This overhead is $\frac{r_c}{r}$ since for every $r$ rounds of the original protocol $r_c$ rounds of verification are used. This fraction of the simulation is lost even if no error occurs. 

The second type of overhead is caused by the errors. To determine this error overhead we note that the adversary can make an iteration useless by investing only one error, e.g., at the beginning of an iteration. The fraction of iterations the adversary can corrupt in this way, without introducing more than an $\eps$ fraction of corruptions, is $\eps (r + r_c)$. Just from this one can already observe that both $r$ and $r_c$ should better be independent of $n$ as $r+r_c$ being anything larger than $1/\eps$ results in the adversary being able to completely stall the algorithm by corrupting every single iteration. If, on the other hand, $r+r_c< 1/\eps$ then the error overhead is  $\frac{1}{1 - \eps (r + r_c)} - 1$ which is at most $\eps (r + r_c)$. 

Putting both the verification and error overhead together leads to the overall communication overhead being at least $\max\{\frac{r_c}{r},\eps (r + r_c)\}$, a tradeoff already familiar from the impossibility result from \Cref{sec:impossibility}. In particular, since $\max\{\frac{r_c}{r},\eps (r + r_c)\} \geq \sqrt{\eps r_c}$ this shows that the only way one can hope to achieve the optimal round complexity of $n (1 + \sqrt{\eps})$ and therefore the rate of $1 - \Theta(\sqrt{\eps})$ is by setting $r = \Theta(\frac{1}{\sqrt{\eps}})$ and $r_c = \Theta(1)$. This is the approach taken by the algorithm presented in this section. 

\subsubsection{Verification using Hashing for String Comparison}

The last  tool needed before we can give a description of the algorithm is an explanation for how the verification is performed. Ideally, the verification would ensure that the transcripts of both parties are in agreement. Of course, this cannot be done in just $r_c = \Theta(1)$ rounds. Instead, we use randomized hashes to at least catch any disagreement with good probability. 

A hash function $h$ can be seen as generating a short fingerprint $h(\Trans)$ for any long string $\Trans$. It is clear that a hash function cannot be injective, that is, for any hash function it is possible to find two strings which have the same fingerprint under $h$. However, good families of hash functions have the nice property that fingerprints of two strings are different with good probability if a random hash function is picked from the family. Typically, a random bit string $\strings$, which is called the \emph{seed}, is used to select this hash function $h_\strings$. Important parameters for families of hash functions are the \emph{seed length}, that is, the amount of randomness needed to select a hash function $h_\strings$, and the probability that a \emph{hash collision} $h_\strings(\stringx) = h_\strings(\stringy)$ happens for two unequal strings $\stringx \neq \stringy$\finalDel{, and the ease with which the fingerprint $h_\strings(\Trans)$ can be computed}. 

Throughout this paper we make use of the following standard hash functions which are derived from the $\eps$-biased probability spaces constructed in \cite{NaorNaor}. These hash functions give the following guarantees:

\begin{lemma}[from \cite{NaorNaor}] \label{lem:hashes}
For any $n$, any alphabet $\Sigma$, and any probability $0<p<1$, there exist $s = \Theta(\log (n \log |\Sigma|) + \log \frac{1}{p})$, $o = \Theta(\log \frac{1}{p})$, and a simple function $h$, which given an $s$-bit uniformly random seed $S$ maps any string over $\Sigma$ of length at most $n$ into an $o$-bit output, such that the collision probability of any two $n$-symbol strings over $\Sigma$ is at most $p$. In short:
\finalOnly{$\forall n,\Sigma,0<p<1: \ \exists s = \Theta(\log (n \log |\Sigma|) + \log \frac{1}{p}), o = \Theta(\log \frac{1}{p}), h: \{0,1\}^s \times \Sigma^n \mapsto \{0,1\}^o\ s.t.$ $\forall \stringx,\stringy \in \Sigma^{\leq n}, \stringx \neq \stringy, \strings \in \{0,1\}^s \ iid\ Bernoulli(1/2): \ \ P[h_\strings(\stringx) = h_\strings(\stringy)] \leq p$}
\finalDel{%
$$\forall n,\Sigma,0<p<1: \ \exists s = \Theta(\log (n \log |\Sigma|) + \log \frac{1}{p}), o = \Theta(\log \frac{1}{p}), h: \{0,1\}^s \times \Sigma^n \mapsto \{0,1\}^o\ s.t.$$
$$\forall \stringx,\stringy \in \Sigma^{\leq n}, \stringx \neq \stringy, \strings \in \{0,1\}^s \ iid\ Bernoulli(1/2): \ \ P[h_\strings(\stringx) = h_\strings(\stringy)] \leq p$$
}
\end{lemma}

The setting of interest to our first algorithm is that any two $\Theta(n)$ bit strings can, with high probability, be distinguished by a random hash function which creates $o = \Theta(\log n)$ size fingerprints and requires a seed of only $\Theta(\log n)$ bits. This allows us to communicate both the selected hash function (seed) and the corresponding hash value of the $O(n)$ long transcript using only $r_c = \Theta(1)$ symbols of $\Theta(\log n)$ bits each.

\subsubsection{Adding Confirmation Steps}

A last minor technical detail is a simple preprocessing step in which the original protocol $\Pi$ is modified by adding $\Theta(\sqrt{\eps} n)$ steps at the end, in which both parties send each other a fixed symbol. These steps should be interpreted as confirmation steps which reaffirm the correctness of the previous conversation. This ensures that the simulation $\Pi'$ never runs out of steps of $\Pi$ to simulate. \fullOnly{It furthermore makes sure that an adversary cannot simply let both parties first compute $\Pi$ correctly and then wait for the end of $\Pi'$ to add a few errors which make both parties backtrack and end in an intermediate step of $\Pi$. With the extra steps the protocol $\Pi'$ is guaranteed to end up in a confirmation step of a valid conversation outcome of $\Pi$. This outcome can then safely be selected as an output of the simulation $\Pi'$.}

}

\subsection{Algorithm}

Putting these ideas together leads to our first, simple coding scheme which achieves the optimal $1 - \Theta(\sqrt{\eps})$ error rate for any logarithmic bit-sized alphabet:

\begin{algorithm}[htb!]
\caption{Simple Coding Scheme for $\Theta(\log n)$-bit alphabet $\Sigma$}
\begin{algorithmic}[1]
\algorithmfontsize

\Statex
\State $\Pi \gets$ $n$-round protocol to be simulated + final confirmation steps
\State $hash \gets$ hash family\finalDel{ from \Cref{lem:hashes}} with $p = 1/n^5$ and $o = s = \Theta(\log n)$    \label{algline:simple:hashdef}

\Statex

\finalDel{\State Initialize Parameters: $r_c \gets \Theta(1)$; \ $r \gets \ceil{\sqrt{\frac{r_c}{\eps}}}$; \ $R_{total} \gets \ceil{n / r + 32n\eps}$; \ $\Trans \gets \emptyset$}
\finalOnly{\State $r_c \gets \Theta(1)$; $r \gets \ceil{\sqrt{\frac{r_c}{\eps}}}$; $R_{total} \gets \ceil{n / r + 32n\eps}$; $\Trans \gets \emptyset$}

\Statex

\For {$R = 1$ to $R_{total}$}
	
	\Statex
	\State $\strings \gets s$ uniformly random bits 		\label{line:simplerandsample}	\Comment{{\bfseries Verification Phase}}
	\State Send $(\strings,hash_\strings(\Trans),|\Trans|)$; Receive $(\strings',H'_{\Trans},l')$
	\State $H_{\Trans} \gets hash_{\strings'}(\Trans)$; $l \gets |\Trans|$

	\Statex
	
	\If{$H_{\Trans} = H_{\Trans}'$}				 \Comment{{\bfseries Computation Phase}}
		\State continue computation of $\Pi$ for $r$ communications and record those in $\Trans$
	\Else 
		\State do $r$ dummy communications keeping $\Trans$ unchanged
	\EndIf

	\Statex

	\If{$H_{\Trans} \neq H_{\Trans}'$ and $l \geq l'$}	\label{line:comparison} \Comment{{\bfseries Transition Phase}}
			\State rollback computation of $\Pi$ and transcript $\Trans$ by $r$ steps
		\EndIf

\EndFor	
	
\Statex
\State Output the outcome of $\Pi$ corresponding to transcript $\Trans$	
	
\end{algorithmic}
\label{alg:SimpleCompute}
\end{algorithm}

\subsection{Proof of Correctness}\label{sec:simpleproof}


\begin{theorem}\label{lem:simplealg}
Suppose any $n$-round protocol $\Pi$ using an alphabet $\Sigma$ of bit-size $\Theta(\log n)$. \refalgSimpleCompute is a computationally efficient randomized coding scheme which given $\Pi$, with probability $1 - 2^{-\Theta(n \eps)}$, robustly simulates it over any fully adversarial error channel with alphabet $\Sigma$ and error rate $\eps$. The simulation uses $n (1 + \Theta(\sqrt{\eps}))$ rounds and therefore achieves a communication rate of $1 - \Theta(\sqrt{\eps})$.
\end{theorem}

\paragraph{Proof Outline}
We show that the algorithm terminates correctly by defining an appropriate potential $\Phi$. We prove that any iteration without an error or hash collision increases the potential by at least one while any error or hash collision reduces the potential by at most some fixed constant. Lastly, we show that with very high probability the number of hash collisions is at most $O(\eps n)$ and therefore negligible. 
 This guarantees an overall potential increase that suffices to show that the algorithm terminates correctly after the fixed $R_{total}$ number of iterations.

\paragraph{Potential}
The potential $\Phi$ is based on the transcript $\Trans$ of both parties. We use $\Trans_A$ and $\Trans_B$ to denote the transcript of Alice and Bob respectively. We first define the following intermediate quantities: The agreement $l^{+}$ between the two transcripts at Alice and Bob is the number of blocks of length $r$ in which they agree, that is, %
\finalDel{$$l^{+} = \floor{\frac{1}{r} \max \left\{l' \in [1,\min\{|\Trans_A|,|\Trans_B|\}] \ \text{s.th.}\  \Trans_A[1, l'] = \Trans_B[1, l']\right\}}.$$}
\finalOnly{$l^{+} = \floor{\frac{1}{r} \max \left\{l' \ \text{s.th.}\  \Trans_A[1, l'] = \Trans_B[1, l']\right\}}.$}
Similarly, we define the amount of disagreement $l^{-}$ as the number of blocks they do not agree on:
\finalDel{$$l^{-} = \frac{|\Trans_A|+|\Trans_B|}{r} - 2 l^+.$$}
\finalOnly{$l^{-} = \frac{|\Trans_A|+|\Trans_B|}{r} - 2 l^+.$}
The potential $\Phi$ is now simply defined as $\Phi = l^+ - l^-$.


\finalDel{\paragraph{Proofs}}

\begin{corollary}\label{lem:simplenoerrorpotential}
Each iteration of \refalgSimpleCompute without a hash collision or error increases the potential $\Phi$ by at least one.
\end{corollary}
\begin{proof}
If $\Trans_A = \Trans_B$ then both parties continue computing $\Pi$ from the same place and, since no error happens, both parties correctly add the next $r$ communications of $\Pi$ to their transcripts. This increases $l^{+}$ and therefore also the overall potential $\Phi$ by one. 

If $\Trans_A \neq \Trans_B$ and no hash collision happens then both parties realize this discrepancy and also learn the correct length of the other party's transcript. If $|\Trans_A|=|\Trans_B|$ then both parties backtrack one block which reduces $l^{-}$ by two and thus increases the potential by two. Otherwise, the party with the longer transcript backtracks one block while the other party does not change its transcript. This reduces $l^{-}$ by one and increases the overall potential $\Phi$ by one.
\end{proof}

\begin{corollary}\label{lem:simpleerrorpotential}
Each iterations of \refalgSimpleCompute, regardless of the number of hash collisions and errors, decreases the potential $\Phi$ by at most three.
\end{corollary}
\begin{proof}
No matter what is received during an iteration a party never removes more than one block from its transcript. Similarly, at most one block is added to $\Trans_A$ and $\Trans_B$. Overall in one iteration this changes $l^+$ by at most by one and $l^{-}$ at most by two. The overall potential $\Phi$ changes therefore at most by three in any iteration. 
\end{proof}

Next, we argue that the number of iterations of \refalgSimpleCompute with a hash collision is negligible. To be precise, we say an iteration \emph{suffers a hash collision} if $\Trans_A \neq \Trans_B$ but either $hash_{\strings_B}(\Trans_A) = hash_{\strings_B}(\Trans_B)$ or $hash_{\strings_A}(\Trans_A) = hash_{\strings_A}(\Trans_B)$. In particular, we do not count iterations as suffering a hash collision if the random hash functions sampled would reveal a discrepancy but this detection, e.g., in \Cref{line:comparison}, is prevented by corruptions in the transmission of a hash value or seed. To prove the number of hash collisions to be small we crucially exploit the fact that the randomness used for hashing is sampled afresh in every iteration. In particular, it is sampled after everything that is hashed in this iteration is already irrevocably fixed. This independence allows to use the collision resistance property of \Cref{lem:hashes} which shows that any iteration suffers from a hash collision with probability at most $p = 1/n^5$. %
\finalDel{%
A union bound over all iterations then shows that with high probability no hash collision happens at all:

\begin{corollary}\label{lem:simplehashcollisionsunionbound}
With probability $1 - 1/n^4$ no iteration in \refalgSimpleCompute suffers from a hash collision.
\end{corollary}

While the success probability guaranteed by \Cref{lem:simplehashcollisionsunionbound} is already quite nice it will be important for subsequent algorithms to realize that one can easily prove stronger guarantees. In particular, using the independence between iterations in combination with a standard tail bound proves that the probability of having a number of hash collisions of the same order of magnitude as the number of errors is at least $1 - 2^{-\Theta(\eps n)}$:

\begin{corollary}\label{lem:simplehashcollisions}
The number of iterations of \refalgSimpleCompute suffering from a hash collision is at most $6n\eps$ with probability at least $1 - 2^{-\Theta(\eps n)}$.
\end{corollary}
\begin{proof}
The hash function family selected in \Cref{algline:simple:hashdef} of \refalgSimpleCompute has, by \Cref{lem:hashes}, the guarantee that a hash collision happens with probability at most $p = 1/n^5$ if the randomness is chosen independently from the strings to be compared. This is the case for \refalgSimpleCompute since the randomness used in an iteration is sampled afresh in \Cref{line:simplerandsample}
right before it is used for hashing. Therefore, even if the adversary influences the transcripts such that the collision probability is maximized, the probability of a collision in any iteration is at most $p$. This remains true even if the adversary learns the random seeds right after they are sampled (which it does since they are sent over the channel). Since at this point the transcripts $\Trans_A$ and $\Trans_B$ to be hashed are fixed and unaffected by any corruptions in this round the adversary does not have any influence on whether the iteration is counted as having suffered a hash collision. Overall, there are at most $2n$ iterations, each with its own independently sampled random seeds. The occurence of hash collisions is thus dominated by $2n$ independent Bernoulli($p$) variables. A Chernoff bound now shows that the probability of having more than $6 n \eps$ hash collisions is at most $p^{\Theta(\eps n)}$.
\end{proof}
} \finalOnly{%
A union bound over all iterations then shows that with high probability no hash collision happens at all. Furthermore, using the independence between iterations in combination with a standard tail bound proves that the probability of having a number of hash collisions of the same order of magnitude as the number of errors is at least $1 - 2^{-\Theta(\eps n)}$:

\begin{corollary}\label{lem:simplehashcollisions}
The number of iterations of \refalgSimpleCompute suffering from a hash collision is at most $6n\eps$ with probability at least $1 - 2^{-\Theta(\eps n)}$.
\end{corollary}
}

We are now ready to prove \Cref{lem:simplealg}:

\begin{proof}[Proof of \Cref{lem:simplealg}]
There are at most $2n\eps$ errors and according to \Cref{lem:simplehashcollisions} at most $6n\eps$ iterations with a hash collision. This results in at most $8n\eps$ iterations in which, according to 
\Cref{lem:simpleerrorpotential} the potential $\Phi$ decreases (by at most three). For the remaining $R_{total} - 8n\eps = \ceil{n / r + 24 n\eps}$ iterations \Cref{lem:simplenoerrorpotential} shows that the potential $\Phi$ increases by one. This leads to a total potential of at least $\ceil{n / r}$ which implies that after the last iteration both parties agree upon the first $n$ symbols of the execution of $\Pi$. This leads to both parties outputing the correct outcome and therefore to \Cref{lem:simplehashcollisions} being a correct robust simulation of $\Pi$. 

To analyze the round complexity and communication rate we note that each of the $R_{total}$ iteration consists of $r$ computation steps and $r_c = \Theta(1)$ symbols exchanged during any verification phase. The total round complexity of \refalgSimpleCompute is therefore 
$R_{total} \cdot (r + r_c) = \ceil{n / r + 6 n\eps} \cdot (r + \Theta(1)) = n + \Theta(n \eps r + n/r + n\eps) = n (1 + \Theta(\eps r + 1/r))$. Choosing the optimal value of $r = \Theta(\frac{1}{\sqrt{\eps}})$, as done in \refalgSimpleCompute, leads to $n(1 + \Theta(\sqrt{\eps})$ rounds in total, as claimed.
\end{proof}

\section{Problems and Solutions for Small Alphabets}\label{sec:problemsandsolutions}

In this section we explain the barriers preventing \refalgSimpleCompute to be applied to channels with small alphabets and then explain the solutions and ideas put forward in this work to circumvent them.

\subsection{Problems with Small Alphabets}

It is easy to see that the only thing that prevents \refalgSimpleCompute from working over a smaller alphabet is that the verification phase uses $\Theta(\log n)$ bits of communication which are exchanged using $r_c = \Theta(1)$ symbols from the large alphabet. For an alphabet of constant size this is not possible and $r_c = \Theta(\log n)$ rounds of verification would be needed\finalOnly{ which is not possible}. \finalDel{ However, as already explained in \Cref{sec:settingrandrc}, a coding scheme in which $r_c$ increases with $n$ is impossible.} In \refalgSimpleCompute the logarithmic amount of communication is used thrice: for the hash function seed, the hash function value, and to coordinate the backtracking by communicating the transcript length.

\subsubsection{Logarithmic Length Information to Coordinate Backtracking}\label{sec:prob:backtrack}

The simplest idea for backtracking would be to have both parties go back some number of steps whenever a non-matching transcript is detected. This works well if both parties have equally long transcripts. Unfortunately, transcripts of different length are unavoidable because the adversary can easily make only one party backtrack while the other party continues. Then, if both parties always backtrack the same number of steps, both parties might reverse correct parts of the transcript without getting closer to each other. This means that with transcripts of different length the parties need to first and foremost come to realize which party is ahead and thus has to backtrack to the transcript length of the other party. Unfortunately, an adversarial channel can easily create transcript length differences of $\Theta(n\eps)$ steps between the two parties. For this it completely interrupts the communication of a simulation, as an ``attacker in the middle'', at a given point of time and simulates a faulty party with Alice, making her backtrack, while simulating a fully compliant party with Bob, making him go forward in an arbitrary wrong direction. With such large length differences it seems hard to achieve synchronization without sending logarithmic size length information.

Another problem is that, especially when dealing with adversarial channels, performing large re-synchronization steps is dangerous. In particular, if there is a way to make a party backtrack for a super-constant number of iterations triggered by only a constant amount of communication then the adversary can exploit this mechanism by faking this trigger. This would lead to $\omega(1)$ backtracking steps for every constant number of errors invested by the adversary and therefore make an optimal communication rate impossible.

\subsubsection{Logarithmic Length Seeds and Hash Values}

In the implementation of \refalgSimpleCompute the seeds used to initialize the hash functions as well as the generated hash value itself are $\Theta(\log n)$ bits long. Looking at \Cref{lem:hashes} reveals that this requirement comes both from the desire to make the collision probability small but also from the fact that we are hashing whole transcripts which are $O(n)$ bits long. One way to try to get around the later problem is to try hashing only the last few rounds. However, in the adversarial setting, it is unavoidable to have errors that go undetected for $n \eps$ rounds since the adversary can completely take over the conversation for this long. Another option would be to look for hash functions with a sub-logarithmic dependence on the length of the strings to be hashed. \finalOnly{Unfortunately, this is not possible (see \cite{H14arxiv}).}\finalDel{
Unfortunately, the next lemma gives a simple argument that this is not possible (unless one uses hash values with almost linear bit-size in which case, e.g., the identity is a good collision free hash function requiring no seed):

\begin{lemma}\label{lem:seedlengthLB}
Any hash function with non-trivial collision probability $p<1$ requires that the seed length $s$ is at least $\log \frac{n \log |\Sigma|}{o}$ where $n \log |\Sigma|$ is of the bit-lengths of the strings it hashes and $o$ is the bit-length of the output.
\end{lemma}

\newcommand{\proofseedlengthLB}{
\begin{proof}\shortOnly{[Proof of \Cref{lem:seedlengthLB}]}
For any seed $s$ the hash function $h_s$ partitions the $2^{n \log |\Sigma|}$ many strings into $2^o$ partitions according to its output value. By pigeon hole principle there is a group of at least $2^{n \log |\Sigma|} 2^{-o}$ strings which evaluate the same given the lexicographically first seed. Repeating this argument gives that there is a group of at least $2^{n \log |\Sigma|} (2^{-o})^i$ strings which evaluate the same under the $i$ lexicographically first seeds. Since there are exactly $2^s$ possible seeds there is a group of at least $2^{n \log |\Sigma|} (2^{-o})^i$ strings which evaluate the same under all seeds. Since two different strings that evaluate the same under all seeds would lead to a trivial collision probability of $1$ we have that 
$$2^{n \log |\Sigma|} (2^{-o})^{2^s} = 2^{n \log |\Sigma| -o 2^s} \leq 1$$
which implies $n \log |\Sigma| - o 2^s \leq 0$ and therefore $2^s \geq \frac{n \log |\Sigma|}{o}$ as claimed. 
\end{proof}
}\fullOnly{\proofseedlengthLB}
\shortOnly{The proof of \Cref{lem:seedlengthLB} is given in\finalDel{ \Cref{app:proofs}}\finalOnly{~\cite{H14arxiv}}.}
}

\subsection{Our Solutions}

In this section we explain our solutions to the above problems and introduce the working parts and rationale behind \refalgComputeOblivious and \refalgComputeAdv.

\subsubsection{Meeting Point Based Backtracking}

We first explain how our algorithms coordinate their backtracking actions while exchanging only $O(1)$ bits per verification phase. In particular, we explain how the parties in our coding scheme determine where to and when to backtrack once they are aware that their transcripts are not in agreement or not synchronized. 

Recall the second observation from \Cref{sec:prob:backtrack} that parties cannot backtrack for more than a constant number of steps for every verification step, which consists of $O(1)$ bits of communication. 
In order to achieve this it is clear that parties might not be able to backtrack at all, even if non-matching transcripts were detected. Our algorithms implement this by maintaining at each party a threshold $\verif$ for how far the party is willing to backtrack. For every iteration with an unresolved transcript inconsistency, that is, for every iteration since the last computation or backtracking step, this threshold increases by one without the party actually performing a backtracking step. The $\verif$ values are kept synchronized between the two parties by including hashes of them in the verification phase and resetting them if too many discrepancies, measured by the error variable $\error$, are observed.  

Now, with both parties having the same threshold $\verif$ for how far they are able and willing to backtrack we use an idea that goes back to Schulman's first interactive coding paper~\cite{Schulman92} and define \emph{meeting points} at which the parties can meet without having to communicate their position, i.e., their $\Theta(\log n)$ bit transcript length description. For this we create a \emph{scale} $\VP$ by rounding $\verif$ to the next power of two, that is, $\VP = 2^{\floor{\log_2 \verif}}$, and define the meeting points on this scale to be all multiples of $\VP r$. A party is willing to backtrack to either of the two closest such meeting points, namely, $\MPone = \VP r \floor{\frac{|\Trans|}{\VP r}}$ and $\MPtwo = \MPone - \VP r$. It is easy to see that these meeting points are consistent and have the property that any two parties with the same scale $\VP$ and a difference of $l^{-} < 2\VP$ have at least one common meeting point up to which their transcripts agree. In each verification phase both parties send hash values of their transcripts up to these two meeting points, in the hope to find a match. We note that for a scale $\VP$ there are $0.5 \VP$ hash comparisons generated during the time both parties look for a common meeting point at this scale $\VP$. If most of these hashes, e.g., $0.4 \VP$ many, indicate a match a party backtracks to this point. This guarantees that producing an incorrect backtracking step of length $d$ requires $O(d)$ corruptions. 

A potential function argument very similar to the one given for \refalgSimpleCompute in \Cref{sec:simpleproof}, except for obviously involving many more cases, shows that this backtracking synchronization works as well as before while communicating only small hashes instead of logarithmic bit-sized length information.

\subsubsection{Hash Values and Seeds}

Next, we explain the strategies we use to reduce the communication overhead in the verification phase stemming from large hash values and seeds. \shortOnly{The discussion regarding the seed length can be found in \finalOnly{\cite{H14arxiv}}\finalDel{\Cref{app:reducingtheseedlength}}. The reader is however highly encouraged to read it before trying to understand the hashing analysis.}

\paragraph{Constant Size Hash Values}
We \finalDel{first} concentrate on reducing the size of the hash values to a constant.
\finalDel{

We begin with the observation that one can get easily away with $\Theta(\log \frac{1}{\eps})$ size hash values. In particular, setting the hash collision probability $p$ from $1/n^5$ to $\eps/2$ results in the expected number of hash collisions to be at most $n \eps$ in total. The same tail bound as used in the proof of \Cref{lem:simplehashcollisions} shows then that the probability of having more than $2 n \eps$ many hash collisions is still at most $p^{\Theta(\eps n)}$. However, using $r_c = \Theta(\log \frac{1}{\eps})$ bits for the verification would still lead to a suboptimal communication rate of $1 - \sqrt{H(\eps)}$.

}
What comes to the rescue here is the observation that hashing only makes one-sided errors, that is, it only confuses different strings for equal but never the other way around since hash values of the same string will always match. Since the primary cause for a non-matching transcript is an iteration with an error one would furthermore expect that there are few, say $O(n\eps)$, opportunities for such a hash collision to happen. This would make it possible to have a constant hash collision probability without increasing the number of hash collisions beyond $\Theta(n \eps)$. \finalOnly{This intuition is indeed correct and can be formalized relatively easily, as shown in \cite{H14arxiv}.}\finalDel{The following lemma formalizes this intuition and shows that \refalgSimpleCompute indeed still works as-is if hash values with only constant bit-size are used:


\begin{lemma}\label{lem:simplealgconstanthashvalue}
\refalgSimpleCompute still functions as claimed in \Cref{lem:simplealg} even if the hash function used has collision probability $p = 0.1$ and therefore output length $o = |H_{\Trans}| = \Theta(1)$.
\end{lemma}
\begin{proof}
The only part of the proof of \Cref{lem:simplealg} that needs to be redone are the arguments in \Cref{lem:simplehashcollisions} which show that the number of iterations suffering a hash collision are at most $6 n \eps$ with probability $1 - 2^{-\Theta(n \eps)}$. 

To show this, we first note that an iteration can have a hash collision only when $l^{-} > 0$, that is, when the transcripts of the two parties disagree. We bound the number of rounds in which this is the case by $6 n \eps$. For this we observe that during any iteration $l^{-}$ increases by at most two while any iteration without any error has an independent probability of at least $1-p$ to reduce $l^{-}$ by one, if it is not zero already. Therefore, if there are $x \geq 6 n \eps$ iterations in which $l^{-} > 0$ then $l^{-}$ must have remained the same or increased during at least $x/3$ iterations of which at most $1.1 n \eps < x/5$  can be attributed to iterations with errors. However, having $x/3 - x/5 > 0.13 x$ hash collisions out of $x$ iterations with $l^{-} >0$ gives an empirical average of $0.13$ among $\Theta(n \eps)$ trials which are dominated by independent Bernoulli trials with probability $p = 0.1$. The same tailbound as before shows that the probability for this to happen is at most $2^{-\Theta(n \eps)}$. This shows that the number of rounds with disagreeing transcripts, and therefore also the number of hash collisions, is at most $6 n \eps$, with probability $1 - 2^{-\Theta(n \eps)}$. 
\end{proof}
}

\newcommand{\reducingseedlength}{
\shortOnly{\section{Reducing the Seed Length via Preshared Randomness} \label{app:reducingtheseedlength}}
\fullOnly{\paragraph{Reducing the Seed Length via Preshared Randomness} \label{app:reducingtheseedlength}}
%
Next, we address how one can reduce the communication overhead caused by \refalgSimpleCompute transmitting in each iteration the seeds that select the random hash functions.

Our general strategy to avoid having to share a logarithmic length seed in every verification step is to share a large amount of randomness at the beginning of the protocol and then use and repeatedly reuse this randomness in every verification step. To share randomness Alice privately samples some uniformly random string $R'$, then encodes this string into a good error correcting code of distance at least $4n \eps$, and finally sends this string to Bob. Since the total number of errors is below $2 n \eps$ Bob can decode correctly. This allows Alice and Bob to agree on some shared random string $R'$ of length $l'$ using only $\Theta(l' + n H(\eps))$ rounds of communication (as long as $l' < n$). A detailed description of this Robust Randomness Exchange algorithm is given as \refalgRobustRandomnessExchange.

This idea however runs into several major obstacles: 
\begin{itemize}
	\item Especially when dealing with a fully adaptive adversary, exchanging randomness, that is used for hashing, seems like a bad idea, because it also informs the adversary about this randomness. The adversary can then choose its corruptions according to the randomness used for hashing which makes it possible to cause hash collisions with certainty. In particular, since hash functions cannot be injective it is easy for an adversary to find two strings $\stringx$,$\stringy$ for which $hash_{\strings}(\stringx) = hash_{\strings}(\stringy)$ if the seed $\strings$ is known to the adversary. We deal with this problem by proving a small failure probability for any oblivious adversary and then showing that the number of (oblivious) strategies an adaptive adversary can adaptively pick from is small enough to apply a union bound. This is similar to the derandomization proofs in \cite{BN13,GH13}. 

	\item Sharing enough randomness to generate an independent $\Theta(\log n)$ bit seed for each of the $n/r$ iterations would require $\Omega(n \log n \sqrt{\eps})$ bits of shared randomness. Sharing such large amounts of randomness would require too much communication. Using a smaller amount of independence is also not directly possible because the transcripts to be hashed in an iteration depend non-trivially on the outcome of all prior (up to) $n/r$ hashing steps. This essentially implies that $n/r$-wise independence is required. On the other hand \Cref{lem:seedlengthLB} shows that the seeds length used in one iteration cannot be made smaller than logarithmic in $n$.  
\end{itemize}

What allows us to circumvent this second obstacle is a more direct use of the approximately $k$-wise independent or $\delta$-biased probability spaces of \cite{NaorNaor}. In particular, we use~\cite{NaorNaor} to deterministically stretch the $l'$ iid random bits in $R'$ to a much longer (pseudo-)random string $R$ of $l$ bits that are $\delta$-biased for some small $\delta$ and therefore are statistically indistinguishable from being independent. For such a seed, \cite{NaorNaor} guarantees the following crucial properties:
\begin{itemize}
	\item Any $\delta$-biased probability space is also $\eps$-statistically close to being $k$-wise independent for $\eps = \delta^{\Theta(1)}$ and $k = \Theta(\log \delta)$. 
	\item This $k$-wise independence extends to linearly independent linear tests. This means that the outcome of $k$ linear tests on a $\delta$-biased probability space is $\eps$-statistically close to $k$ fully independent tests as long as the tests are linearly independent. This holds even if each tests compromises a large number of variables.
	\item Lastly, only $\Theta(\log l + \log \delta)$ random bits are required to create $l$ random bits with bias $\delta$. Even for a large, polynomial sized $R$ the amount of shared randomness required to produce $R$ is dominated by the \emph{additive} $\Theta(\log \delta)$ term. This means that the amount of randomness required is only $\Theta(1)$ times the amount of independence required.
\end{itemize}

In \Cref{sec:innerproducthash} we show how to use these properties with a very simple hash function.

%
}
\fullOnly{\reducingseedlength}

\section{Our Coding Schemes}\label{sec:protocol}

\finalDel{Next, we give a complete description of our coding schemes for oblivious and fully adversarial channels. We start with a description of the Randomness Exchange subroutine which is given as \refalgRobustRandomnessExchange and then describe the hash functions we use in our coding scheme. \refalgComputeOblivious then gives the coding scheme for oblivious channels and \refalgComputeAdv is the variation that works for fully adversarial channels.}

\subsection{The Robust Randomness Exchange Subroutine}

The Robust Randomness Exchange Subroutine is used to exchange some randomness at the beginning of the algorithm using an error correcting code which is then stretched to a longer $\delta$=biased pseudo random string of length $l$ using \cite{NaorNaor}. This string is then used by both parties to provide the random seeds for selecting the hash functions in each iteration\finalOnly{ (see also \cite{H14arxiv})}. 

\begin{algorithm}[htb!]
\caption{Robust Randomness Exchange($l$,$\delta$)} 
\begin{algorithmic}[1]
\algorithmfontsize
\State Input: desired number of bits $l$ and bias $\delta$ of the shared randomness
\State Output: shared random string $R$ of length $l$ and bias $\delta$
\Statex
\State $l' = \Theta(\log \delta + \log l)$
\State $C \gets $ \finalDel{Error Correcting Code}\finalOnly{ECC} $\{0,1\}^{l'} \rightarrow \{0,1\}^{\Theta(l' + n H(\eps))}$ with distance $4n\eps$

\finalDel{\Statex}
\If{Alice}		\finalDel{\Comment{{\bfseries Randomness Exchange} requiring $\Theta(\log \delta + \log l + nH(\eps))$ rounds}}
	\State $R' \gets$  uniform random bit string of length $l'$
	\State Transmit $C(R')$ to Bob  \finalDel{\Comment Send Encoding of $R'$ to Bob}
\ElsIf{Bob}
	\State Receive $C'$ from Alice	\finalDel{\Comment Receive Corrupted Codeword from Alice}
	\State $R' \gets$ Decoding of $C'$ 
\EndIf	

\finalDel{\Statex \Comment Generate shared (pseudo-)random string $R$}
\State $R \gets \delta$-biased pseudo random string of length $l$ derived from $R'$ 

\end{algorithmic}
\label{alg:RobustRandomnessExchange}
\end{algorithm}

\subsection{The Inner Product Hash Function}\label{sec:innerproducthash}

In our algorithms we use the following, extremely simple, \emph{inner product hash function}, which allows for an easy analysis given the $\delta$-biased property of the shared random seed:

\begin{definition}[Inner Product Hash Function]\label{def:innerproducthash}
For any input length $L$ and any output length $o$ we define the inner product hash function $h_\strings(.)$ as doing the following:
For a given binary seed $\strings$ of length at least $2oL$ it takes any binary input string $\stringx$ of length $l \leq L$, concatenates this input with its length $\tilde{\stringx} = (\strings,|\strings|)$ to form a string of length $\tilde{l} = |\tilde{\stringx}| = |\stringx| + \ceil{\log_2 |\stringx|} \leq 2L$ and then outputs the $o$ inner products $\left\langle \tilde{\stringx}, \strings[i \cdot 2L + 1,i \cdot 2L +\tilde{l}]\right\rangle$ for every $i \in [0,o-1]$. 
\end{definition}

The next corollary states the trivial fact that the inner product hash function is a reasonable hash function with collision probability exponential in its output length if a (huge) uniformly random seed is used. It also states that replacing this uniform seed by a $\delta$-biased one does not change the outcome much. This follows directly from the definition of $\delta$-bias:

\begin{corollary}
Consider a pairs of binary strings $\stringx \neq \stringy$ each of length at most $L$, and suppose $h$ is the inner product hash function for input length $L$ and any output length $o$. Suppose furthermore that $\strings$ is seed string of length at least $n \cdot 2 o L$  which is sampled independently of $\stringx, \stringy$. The collision probability $P[h_\strings(\stringx) = h_\strings(\stringy)]$ is exactly $2^{-o}$ if $\strings$ is sampled from the uniform distribution. Furthermore, if the seed $\strings$ is sampled from a $\delta$-biased distribution the collision probability remains at most $2^{-o} + \delta$.
\end{corollary}

Lastly, the next lemma summarizes the advantage of the inner product hash function in combination with a $\delta$-biased seed, namely that this bias translates directly to the exact same bias on the output distribution. This uses the above mentioned fact from \cite{NaorNaor} that $\delta$-bias also extends beyond variables to any set of linearly independent tests: 

\begin{lemma}\label{lem:independenceoflinearhash}
Consider $n$ pairs of binary strings $(\stringx_1,\stringy_1),\ldots,(\stringx_n,\stringy_n)$ where each string is of length at most $L$, and suppose $h$ is the inner product hash function for input length $L$ and any output length $o$. Suppose furthermore that $\strings$ is a random seed string of length at least $n \cdot 2 o L$ which is sampled independently of the $\stringx$ and $\stringy$ inputs and is cut into $n$ strings $\strings_1,\strings_2,\ldots,\strings_n$. Then the output distribution $(x_1,\ldots,x_n) = (h_{\strings_1}(\stringx_2) - h_{\strings_2}(\stringy_1), \ldots, h_{\strings_n}(\stringx_2) - h_{\strings_n}(\stringy_1))$ for a $\strings$ sampled from a $\delta$-biased distribution is $\delta$-statistically close to the output distribution for a uniformly sampled $\strings$ for which each $x_i$ is equal to zero if 
$\stringx_i= \stringy_i$ and independently uniformly random otherwise (which also implies $P[x_i = 0] = 2^{-o}$). 
\end{lemma}


\begin{algorithm}[htb!]
\caption{Coding Scheme for Oblivious \finalOnly{Adv.}\finalDel{Adversarial} Channels}
\begin{algorithmic}[1]
\algorithmfontsize
\Statex
\State $\Pi \gets$ $n$-round protocol to be simulated + final confirmation steps
\State $hash \gets $ inner product hash\finalOnly{ fam.}\finalDel{family from \Cref{def:innerproducthash}} with $o = \Theta(1)$ and $s = \Theta(n)$ \label{algline:defhash}

\Statex
\finalOnly{\State $r_c \gets \Theta(1)$; $r \gets \ceil{\sqrt{\frac{r_c}{\eps}}}$; $R_{total} \gets \ceil{n / r + 65 n\eps}$; $\Trans \gets \emptyset$}
\finalDel{\State Initialize Parameters: $r_c \gets \Theta(1)$; $r \gets \ceil{\sqrt{\frac{r_c}{\eps}}}$; \ $R_{total} \gets \ceil{n / r + 65 n\eps}$; \ $\Trans \gets \emptyset$}

\State Reset Status: $\verif,\error,\voteone,\votetwo \gets 0$ \label{line:resetstatusobliv}

\Statex

\State R = Robust Randomness Exchange($l = R_{total} \cdot s$, $\delta = 2^{-\Theta(\frac{n}{r}o)}$) \label{algline:randexchange} \finalDel{\Comment Requires $\Theta(n \sqrt{\eps})$ rounds}

\Statex

\For {$R_{total}$ iterations}

	\Statex				
	\State $\verif \gets \verif + 1$; $\VP \gets 2^{\floor{\log_2 \verif}}$\finalDel{; \ }\finalOnly{\State }$\MPone \gets \VP r \floor{\frac{|\Trans|}{\VP r}}$; $\MPtwo \gets \MPone - \VP r$												
	\Comment{{\bfseries  Verification Phase}}
	\State $\strings \gets s$ new preshared random bits	from $R$					\label{algline:hashingbegin}
	\State Send $(hash_{\strings}(\verif),hash_{\strings}(\Trans),hash_{\strings}(\Trans[1,\MPone]),hash_{\strings}(\Trans[1,\MPtwo]))$
	\State Receive $(H_{\verif}',H_{\Trans}',H_{\MPone}',H_{\MPtwo}')$;
\finalOnly{ 
	\State $h(.) = hash_{\strings}(.)$
	\State $(H_{\verif},H_{\Trans},H_{\MPone},H_{\MPtwo}) \gets {\scriptsize (h(\verif),h(\Trans),h(\Trans[1,\MPone]),h(\Trans[1,\MPtwo]))}$ \label{algline:hashingend}
}\finalDel{
	\State $(H_{\verif},H_{\Trans},H_{\MPone},H_{\MPtwo}) \gets (hash_{\strings}(\verif),hash_{\strings}(\Trans),hash_{\strings}(\Trans[1,\MPone]),hash_{\strings}(\Trans[1,\MPtwo]))$   \label{algline:hashingend} 
	\Statex
	}

	\If{$H_{\verif} \neq H_{\verif}'$}	\label{algline:ComputeAfterHash}
		\State $\error \gets \error + 1$
	\Else
		\If{$H_{\MPone} \in \{H_{\MPone}',H_{\MPtwo}'\}$}
			\State $\voteone \gets \voteone + 1$  \label{line:voteone}
		\ElsIf{$H_{\MPtwo} \in \{H_{\MPone}',H_{\MPtwo}'\}$}
			\State $\votetwo \gets \votetwo + 1$  \label{line:votetwo}
		\EndIf	
	\EndIf	
	
	\Statex																										
	\If {$\verif = 1$ and $H_{\Trans} = H_{\Trans}'$ and $\error = 0$}												\Comment{{\bfseries Computation Phase}}
		\State continue computation and transcript $\Trans$ for $r$ steps
		\State Reset Status: $\verif,\error,\voteone,\votetwo \gets 0$ \label{line:resetcomputation}
	\Else 
		\State do $r$ dummy communications
	\EndIf	

	\Statex				
	\If {$2 \error \geq \verif$}			\label{line:transitionerror}			\Comment{{\bfseries Transition Phase}}
		\State Reset Status: $\verif,\error,\voteone,\votetwo \gets 0$ \label{line:reseterror}
	\ElsIf {$\verif = \VP$ \ \ and \ \ $\voteone \geq 0.4 \cdot \VP$}     \label{line:transitionrollbackone}
		\State rollback computation and transcript $\Trans$ to position $\MPone$   \label{line:rollbackone}
		\State Reset Status: $\verif,\error,\voteone,\votetwo \gets 0$ \label{line:resetrollbackone}
	\ElsIf{$\verif = \VP$ \ \ and \ \ $\votetwo \geq 0.4 \cdot \VP$}			\label{line:transitionrollbacktwo}
		\State rollback computation and transcript $\Trans$ to position $\MPtwo$
		\State Reset Status: $\verif,\error,\voteone,\votetwo \gets 0$ \label{line:resetrollbacktwo}
	\ElsIf{$\verif = \VP$}
		\State $\voteone,\votetwo \gets 0$										\label{line:voteresetdouble}
	\EndIf
\EndFor

\Statex
\State Output the outcome of $\Pi$ corresponding to transcript $\Trans$	    \label{algline:ComputeEnd}

\end{algorithmic}
\label{alg:ComputeOblivious}
\end{algorithm}

\begin{algorithm}[htb!]
\caption{Coding Scheme for Fully Adversarial Channels}
\begin{algorithmic}[1]
\algorithmfontsize
\Statex
\State $\Pi \gets$ $n$-round protocol to be simulated + final confirmation steps
\State $hash_1 \gets $ inner product hash\finalOnly{ fam.}\finalDel{ family from \Cref{def:innerproducthash}} with $o_1 = \Theta(\log \frac{1}{\eps})$ and $s_1 = \Theta(n)$ 
\State $hash_2 \gets $ hash\finalOnly{ fam.}\finalDel{ family from \Cref{lem:hashes}} with $p_2 = 0.1$, $o_2 = \Theta(1)$, and $s_2 = \Theta(\log \log \frac{1}{\eps})$
\Statex
\finalOnly{\State $R_{total} \gets \ceil{n / r} + \Theta(n\eps)$; $r_c \gets \Theta(\log \log \frac{1}{\eps})$; $r \gets \ceil{\sqrt{\frac{r_c}{\eps}}}$; $\Trans \gets \emptyset$}
\finalDel{\State Initialize Parameters: $R_{total} \gets \ceil{n / r} + \Theta(n\eps)$; \ $r_c \gets \Theta(\log \log \frac{1}{\eps})$; \ $r \gets \ceil{\sqrt{\frac{r_c}{\eps}}}$; \ $\Trans \gets \emptyset$}

\State Reset Status: $\verif,\error,\voteone,\votetwo \gets 0$ \label{line:resetstatusadv}

\Statex

\State R = Robust Randomness Exchange($l = R_{total} \cdot s_1$, $\delta = 2^{-\Theta(\frac{n}{r})}$) \finalDel{\Comment Requires $O(n \sqrt{\eps})$ rounds}

\Statex

\For {$R_{total}$ iterations}

	\Statex				
	\State $\verif \gets \verif + 1$; $\VP \gets 2^{\floor{\log_2 \verif}}$\finalDel{; \ }\finalOnly{\State }$\MPone \gets \VP r \floor{\frac{|\Trans|}{\VP r}}$; $\MPtwo \gets \MPone - \VP r$												
	\Comment{{\bfseries  Verification Phase}}
	\State $\strings_1 \gets s_1$ new preshared random bits from $R$\finalDel{; \ }\finalOnly{\State }$\strings_2 \gets  s_2$ ``fresh'' random bits
	\State $hash(.) = hash_{2,\strings_2}(hash_{1,\strings_1}(.))$
	\State Send $(\strings_2,hash(\verif),hash(\Trans),hash(\Trans[1,\MPone]),hash(\Trans[1,\MPtwo]))$\finalDel{; \ }\finalOnly{\State }Receive $(\strings_2',H_{\verif}',H_{\Trans}',H_{\MPone}',H_{\MPtwo}')$;
\finalOnly{%
	\State $h'(.) = hash_{2,\strings_2'}(hash_{1,\strings_1}(.))$
  \State {\scriptsize $(H_{\verif},H_{\Trans},H_{\MPone},H_{\MPtwo}) \gets (h'(\verif),h'(\Trans),h'(\Trans[1,\MPone]),h'(\Trans[1,\MPtwo]))$}
}\finalDel{%
	\State $hash'(.) = hash_{2,\strings_2'}(hash_{1,\strings_1}(.))$
	\State $(H_{\verif},H_{\Trans},H_{\MPone},H_{\MPtwo}) \gets (hash'(\verif),hash'(\Trans),hash'(\Trans[1,\MPone]),hash'(\Trans[1,\MPtwo]))$
}
	\Statex

	\State {\bfseries Remaining Code as in Lines \ref{algline:ComputeAfterHash} to \ref{algline:ComputeEnd} in \refalgComputeOblivious}
	
\EndFor

\end{algorithmic}
\label{alg:ComputeAdv}
\end{algorithm}

\section{Analyses and Proofs of Correctness}\label{sec:proof}

\finalDel{Next, we give the proofs of correctness for \refalgComputeOblivious and for \refalgComputeAdv. }

\begin{theorem}\label{lem:mainoblivious}
Suppose any $n$-round protocol $\Pi$ using any alphabet $\Sigma$. \refalgComputeOblivious is a computationally efficient randomized coding scheme which given $\Pi$, with probability $1 - 2^{-\Theta(n \eps)}$, robustly simulates it over any oblivious adversarial error channel with alphabet $\Sigma$ and error rate $\eps$. The simulation uses $n (1 + \Theta(\sqrt{\eps}))$ rounds and therefore achieves a communication rate of $1 - \Theta(\sqrt{\eps})$.
\end{theorem}

\begin{theorem}\label{lem:mainadv}
Suppose any $n$-round protocol $\Pi$ using any alphabet $\Sigma$. \refalgComputeAdv is a computationally efficient randomized coding scheme which given $\Pi$, with probability $1 - 2^{-\Theta(n \eps)}$, robustly simulates it over any fully adversarial error channel with alphabet $\Sigma$ and error rate $\eps$. The simulation uses $n (1 + \Theta(\sqrt{\eps \log \log \frac{1}{\eps}})$ rounds and therefore achieves a communication rate of $1 - \Theta(\sqrt{\eps \log \log \frac{1}{\eps}})$.
\end{theorem}

\paragraph{Proof Outline}
We use the same proof structure as already introduced in \Cref{sec:simpleproof}. In particular, we show that the algorithm terminates correctly by defining a potential $\Phi$. We prove that any iteration without an error or hash collision increases the potential by at least a constant while any iteration with an error or hash collision reduces the potential at most by some constant. We do this in two steps: First we show that this statement is true for the computation and verification phase of each iteration only. We then show that any transition in the transition phase does not decrease the potential. As a last step, we bound the number of hash collisions to be of the same order as the number of errors. This is the sole part in which the analyses of \refalgComputeOblivious and \refalgComputeAdv differ. 


\paragraph{Potential}
The potential $\Phi$ is based on the variables $\verif$, $\error$, and $\Trans$ of both parties. For these variables we use a subscript $A$ or $B$ to denote the value of the variable for Alice and Bob respectively. We also denote with the subscript $AB$ the sum of both these variables, e.g., $k_{AB} = k_A + k_B$. To define $\Phi$ we need the following intermediate quantities:

As before, we define the amount of agreement $l^{+}$ and disagreement $l^{-}$ between the two paths computed at Alice and Bob as 
\finalDel{$$l^{+} = \floor{\frac{1}{r} \max \left\{l' \in [1,\min\{|\Trans_A|,|\Trans_B|\}] \ \text{s.th.}\  \Trans_A[1, l'] = \Trans_B[1, l']\right\}}\ \ \text{and} \ \ l^{-} = \frac{|\Trans_A|+|\Trans_B|}{r} - 2 l^+.$$}
\finalOnly{$l^{+} = \floor{\frac{1}{r} \max \left\{l' \ \text{s.th.}\  \Trans_A[1, l'] = \Trans_B[1, l']\right\}}$ and $l^{-} = \frac{|\Trans_A|+|\Trans_B|}{r} - 2 l^+.$}

For sake of the analysis we also define two variables $\BVC_A$ and $\BVC_B$ which count the contribution of hash collisions and corruptions to $\voteone$ and $\votetwo$ at Alice and Bob. In any iteration in which $\voteone$ of either party increases in \Cref{line:voteone} without $\Trans[1,\MPone]$ matching either $\Trans[1,\MPone]$ or $\Trans[1,\MPtwo]$ of the other party we count this as a bad vote and increase \emph{both} $\BVC_A$ and $\BVC_B$. Similarly, we increase both $\BVC$ values if $\votetwo$ of a party increases in \Cref{line:votetwo} without $\Trans[1,\MPtwo]$ matching either $\Trans[1,\MPone]$ or $\Trans[1,\MPtwo]$ of the other party. On the other hand, if one such match occurs but the corresponding vote does not increase, e.g., due to a corruption, then we call this an uncounted vote and also increase $\BVC_A$ and $\BVC_B$ by one. With every status reset (\Cref{line:reseterror,line:resetrollbackone,line:resetrollbacktwo}) we also set the $\BVC$ count of this party to be zero. We remark that the $\BVC$ values are not known to either party; they are merely used to facilitate our analysis.


To weight the various contributions to the potential we use the constants 
\finalDel{$$1 < C_2 < C_3 < C_4 < C_5 < C_6,$$}
\finalOnly{$1 < C_2 < C_3 < C_4 < C_5 < C_6,$}
which are chosen such that $C_i$ is sufficiently large depending only on $C_j$ with $j<i$.
\finalDel{

}%
The potential $\Phi$ is now defined to be %
\finalOnly{ \\$ \mbox{}\ \ l^+ - C_3 \cdot l^- \,\, + \,\ \ \;\ C_2 \cdot \verif_{AB} - C_5 \cdot \error_{AB} - \,2 C_6 \cdot \BVC_{AB}$\\ if $\verif_A = \verif_B$ and\\    $\mbox{}\ \ l^+ - C_3 \cdot l^- - 0.9 C_4 \cdot \verif_{AB} + C_4 \cdot \error_{AB} - \,\; C_6 \cdot \BVC_{AB}$ otherwise.}
\finalDel{%
as follows:
\shortOnly{
$\Phi = \twopartdef%
{l^+ - C_3 \cdot l^- + \ \ \;\ C_2 \cdot \verif_{AB} - C_5 \cdot \error_{AB} - 2 C_6 \cdot \BVC_{AB}}{\verif_A = \verif_B}
{l^+ - C_3 \cdot l^- - 0.9 C_4 \cdot \verif_{AB} + C_4 \cdot \error_{AB} - \,\; C_6 \cdot \BVC_{AB}}{\verif_A \neq \verif_B}$
}
\fullOnly{
$$\Phi = \twopartdef%
{l^+ \ - \  C_3 \cdot l^- \ + \ \ \ \;\ C_2 \cdot \verif_{AB} \ - \ C_5 \cdot \error_{AB} \ - \ 2 C_6 \cdot \BVC_{AB}}{\verif_A = \verif_B}
{l^+ \ - \  C_3 \cdot l^- \ - \ 0.9 C_4 \cdot \verif_{AB} \ + \ C_4 \cdot \error_{AB} \ - \ \,\; C_6 \cdot \BVC_{AB}}{\verif_A \neq \verif_B}$$
}}

\finalDel{\paragraph{Proofs}}

\begin{lemma}\label{lem:potentialincrease1}
In every computation and verification phase the potential decreases at most by a fixed constant, regardless of the number of errors and hash collisions. Furthermore, in the absence of an error or hash collision  the potential strictly increases by a at least one. 
\end{lemma}
\begin{proof}
All quantities on which the potential depends change at most by a constant during any computation and verification phase. The maximum potential change is therefore at most a constant. Now we consider the case that no error or hash collision happened. In this case, the $\BVC_{AB}$ value does not change. Furthermore, computation only happens if $\Trans_A = \Trans_B$ which implies that the $l^+ - C_3 \cdot l^-$ part of the potential does not decrease. Lastly, both $\verif_A$ and $\verif_B$ increase by one and if they are not equal $\error_{A}$ and $\error_{B}$ increase by one, too. In the first case the increase of $\verif_{AB}$ leads to a total potential increase of $2C_2 > 1$ in the later case the increase of $\verif_{AB}$ and $\error_{AB}$ leads to a total potential change of $2(- 0.9 C_4 + C_4)$ which is at least one for sufficiently large $C_4$. The potential therefore strictly increases by at least one in the computation and update phase when no error or hash collision happens.
\end{proof}

\begin{lemma}\label{lem:potentialincreasetwo}
In every iteration the potential decreases at most by a fixed constant, regardless of the number of errors and hash collisions. Furthermore, in the absence of an error or hash collision the potential strictly increases by at least one. 
\end{lemma}

\newcommand{\proofpotentialincreasetwo}{
\shortOnly{\begin{proof}[Proof of \Cref{lem:potentialincreasetwo}]}\fullOnly{\begin{proof}}
Given \Cref{lem:potentialincrease1} it suffices to show that a transition phase never decreases the potential. We show exactly this, except for one case, in which the potential decreases by a small constant. In case of the iteration being error and hash collision free this constant is shown to be less than the increase of the preceding computation and verification phase.

We call a transition due to \Cref{line:transitionerror} an Error Transition and any transition due to \Cref{line:transitionrollbackone} or \Cref{line:transitionrollbacktwo} a Meeting Point Transition. We denote with $l^+$,$l^-$,$\verif_A$,$\error_A$,$\BVC_A$,$\verif_B$,$\error_B$,$\BVC_B$ the values before any transition and denote with $l'^+$,$l'^-$,$\verif_A'$,$\error_A'$,$\BVC_A'$,$\verif_B'$,$\error_B'$,$\BVC_B'$ the values after the transition. We also use a $\Delta$ in front of any variable to denote the change of value to this variable during the transition phase. 

We now make the following case distinction according to which combination of transition(s) occurred in the iteration at hand and whether or not the parties agreed in their $\verif$ parameter before the transition:

\begin{itemize}

\item If $\verif_A \neq \verif_B$ and exactly one error or meeting point transition occurred then the disagreement in the $\verif$ parameter remains, that is, $\verif_A' \neq \verif_B'$. We assume, without loss of generality, that Alice does the transition. The bad vote count $\BVC_{AB}$ never increases in an error and collision free iteration so any change in $\BVC_{AB}$ only increases the potential. Furthermore, since the error counts increase at most by one per round and since after any round $2\error < \verif$ holds we have $\error_A \leq \verif_A/2 + 0.5$. Lastly, all quantities measured in the potential change by at most $\verif_A$. This leads to an overall potential increase of at least $(0.9 C_4 - 0.5 C_4 - C_3 - 1)\verif_A - 0.5 C_4$. This is larger than one for $\verif_A > 1$ and a sufficiently large $C_4$. For $\verif_A=1$ the difference can be negative. However, in this case, Alice had a reset status in the directly preceding iteration and is in the same position except for possibly a one block shorter transcript at the end of this iteration while Bob increased $\error_{B}$ and $\verif_{B}$. This leads to an overall potential increase of at least $0.1 C_4 - C_3 - 1$ which is at least one for a sufficiently large $C_4$.


\item If $\verif_A \neq \verif_B$ and any two transitions occurred then $\verif_A' = \verif_B' = \BVC_{A}' = \BVC_B' = 0$ which makes the potential exactly $l'^+ - C_3 \cdot l'^-$ after the transition. The contribution of $\Delta\BVC_{AB} \leq 1$ only increases the potential and it therefore suffices to show that the contributions of $\Delta\verif_{AB}$ and $\Delta\error_{AB}$ are positive. As before we have $\error_A \leq \verif_A/2 + 0.5$ and $\error_B \leq \verif_B/2 + 0.5$ and therefore also $\error_{AB} \leq 0.5 \verif_{AB} + 1$. The overall potential increase is thus at least $0.9 C_4 \verif_{AB} - C_4 \error_{AB} \geq C_4 (0.9 \verif_{AB} - (0.5 \verif_{AB} + 1)) = C_4 (0.4 \verif_{AB} - 1) \geq 1$ where the last inequality follows for large enough $C_4$ because $\verif_{AB} \geq 3$.

\item If $\verif_A = \verif_B$ and at least one error transitions occurred then the potential change is dominated by the reduction or re-weighting of the $\BVC_{AB}$ count or by the reduction in the $\error$ variable(s) of the transitioning party or parties. In particular, the $\BVC$ of a party does not increase and is never weighted higher before the transition than after. Any $\BVC$ change therefore only affects the potential positively. The $\error$ count of the party with the error transition on the other hand is at least $0.5 \verif$ which leads to a potential change of at least $0.5 \verif C_5$. All other quantities influencing the potential are changing at most by $3\verif$ while being associated with smaller constants. This guarantees an overall potential change of at least one for a sufficiently large $C_5$.  

\item If $\verif_A = \verif_B$ and one meeting point transitions occurred alone then both parties had the same $\verif$ count for the last $\verif$ iterations and either both or none of the parties should have transitioned. This guarantees that the $\BVC$-count of the transitioning party is at least $\verif/4$ which guarantees that the overall potential change is dominated by the $C_6 \verif /4$ increase due to the new $\BVC$ count. 

\item The last case is $\verif_A = \verif_B$ with two meeting point transitions. If $l'^{-} \neq 0$ or if $\verif_A = \verif_B \geq 4 l^{-}$ then both parties should have not transitioned or transitioned to a common meeting point since over $l^{-}$ iterations. Both situations can therefore only arise if $\BVC_{AB} \geq 0.4 l^{-}$ in which case the potential increase due to the reduction in the $\BVC$ count dominates any other potential changes. If, on the other hand, $l'^{-} = 0$ and $\verif_A = \verif_B \leq 4 l^{-}$ then the total potential difference is at least $C_3 l^- - (C_2 + 1) \verif_{AB} \geq (C_3/4 - 2 C_2 - 2) \verif$ and therefore at least one for large enough $C_3$. 

\end{itemize}

This completes the proof for an overall increase in potential for any iteration in which no error or hash collision occurred while for iterations with an error or hash collision at most a constant drop in potential can come from the verification and computation phase.  
\end{proof}
}\fullOnly{\proofpotentialincreasetwo}

Next, we show that our hash function families and the randomness used in our algorithms is strong enough to ensure that the total number of hash collisions is small, namely comparable to the number of errors. This allows us to treat any iteration with a hash collisions as adversarially corrupted and thus equivalent to an iteration with errors. 

We start by showing that the potential $\Phi$ cannot grow too fast. In particular, it grows naturally by one per iteration when a correct computation step is performed. On the other hand, any corruption cannot increase this by more than a constant per error:

\begin{lemma}\label{lem:totalpotential}
The total potential $\Phi$ after $R$ iterations is at most $R + 20 C_2 n \eps$. 
\end{lemma}
\newcommand{\prooftotalpotential}{
\shortOnly{\begin{proof}[Proof of \Cref{lem:totalpotential}]}\fullOnly{\begin{proof}}
We have that $l^{+}$ increases at most by one in each iteration which gives $l^{+} \leq R$. Furthermore, at the end of an iteration it always holds that $2 \error < \verif$  which implies that $-0.9 C_4 \verif_{AB} + C_4 \error_{AB} \leq - 0.4 \verif_{AB} \leq 0$ which leads to the potential $\Phi$ being at most $R$ if $\verif_A \neq \verif_B$. The only way to have a potential larger than $R$ is therefore if $\verif_A = \verif_B$ and $\verif_{AB}$ is large compared to $l^{-}$. This however is not possible without a large number of errors. More precisely, in order to have $\Phi \geq R + x$ it has to be true that $x \leq C_2 \verif_{AB} - C_3 l^{-}$ or $x \leq C_2 (\verif_{AB} - 2l^{-})$ assuming $C_3 > 2 C_2$. However, the only way for $\verif_{AB}$ to be larger than $2 l^{-}$ is if at least ten percent of the votes in the last $\verif_{AB} - 2 l^{-}$ rounds were corrupted to appear non-matching. These kind of corruptions can furthermore not be caused by a hash collision and therefore must be due to an error which implies that $10\% \cdot (\verif_{AB} - 2 l^{-}) < 2 n \eps$ or $(\verif_{AB} - 2 l^{-}) < 20 n \eps$. Putting this together gives $x \leq 20 C_2 n \eps$ and therefore $\Phi \leq R + 20 C_2 n \eps$ as desired. 
\end{proof}
}\fullOnly{\prooftotalpotential}

Now we can show the number of hash collisions in \refalgComputeOblivious to be small:

\begin{lemma}\label{lem:hashcollisionoblivious}
For any protocol $\Pi$ and any oblivious adversary the number of iterations suffering a hash collision in \refalgComputeOblivious is at most $\Theta(\eps n)$, with probability $1 - 2^{-\Theta(n \eps)}$.
\end{lemma}

\newcommand{\proofhashcollisionoblivious}{
\shortOnly{\begin{proof}[Proof of \Cref{lem:hashcollisionoblivious}]}\fullOnly{\begin{proof}}
We call an iteration \emph{dangerous} if, at the beginning of the iteration, the states of both parties do not agree, that is, if either $l^{-} > 0$ or $\verif_{AB} > 0$. It is clear that hash collisions can only occur during dangerous iterations, hence their name. Let $d$ be the number of dangerous iterations and let $h$ be the number of hash collisions (during these iterations). We want to prove that for $h = \Theta(n \eps)$ sufficiently large the fraction $\frac{h}{d}$ of hash collisions in dangerous rounds is too large, in particular, much larger than the expected fraction of hash collisions $p = 2^{-o}$, where the output length $o$ is chosen sufficiently large. In \refalgComputeOblivious we have $\log \delta = \Theta(R_{total})$ so large that according to \Cref{lem:independenceoflinearhash} all $R_{total} = \Theta(n \sqrt{\eps})$ hashing steps are statistically close to being fully independent. In particular, hash collisions are statistically close to being dominated by independent Bernoulli($p$) trials. A tail bound then shows that the probability for having such a large deviation from the expectation is exponentially small in $h$:

\Cref{lem:potentialincreasetwo} shows that in each iteration the potential increases at least by a fixed constant $C^{+}$ if no error or hash collision happens while it decreases at most by a fixed constant $C^{-}$ otherwise. The total potential change during dangerous rounds is therefore at least $C^{+} (d - h - 2 \eps n) - C^{-} (h + 2 n \eps) $ while the potential accumulated in non-dangerous rounds is at least $R_{total} - C^{-} (2n\eps)$ for a total potential of at least $R_{total} + C^{+} (d - h - 2 \eps n) - C^{-} (h + 4 n \eps)$. From \Cref{lem:totalpotential} we get however that the total potential is at most $R_{total} + 20 C_2 n \eps$. Together this implies $C^{+} (d - h) - C^{-} h = O(n \eps)$ and therefore also $d \leq (\frac{C^{-}}{C^{+}} + 1) h + \Theta(n \eps)$.
This implies that if $d = \Theta(n \eps)$ is sufficiently large then $h$ needs to be at least $\frac{1}{2(\frac{C^{-}}{C^{+}} + 1)} d$ and if we choose $p = 2^{-o} < \frac{1}{4(\frac{C^{-}}{C^{+}} + 1)}$ 
the probability for having this large of a fraction of collisions during dangerous rounds becomes an arbitrarily small $2^{-\Theta(n \eps)}$. Therefore the number of dangerous rounds is at most $d = \Theta(n \eps)$ with probability $1 - 2^{-\Theta(n \eps)}$. This also implies that the number of hash collisions is $\Theta(n \eps)$ as desired.
\end{proof}
}\fullOnly{\proofhashcollisionoblivious}

Next, we show that the $hash_1$ hash function in \refalgComputeAdv also causes at most $\Theta(\eps n)$ hash collisions, even for a fully adversarial channel:

\begin{lemma}\label{lem:hashonecollisions}
For any protocol $\Pi$ and any fully adversarial channel the number of iterations with hash collisions due to first hashing with the hash function $hash_1$ in \refalgComputeAdv is at most $\Theta(\eps n)$, with probability $1 - \eps^{\Theta(n \eps)}$.
\end{lemma}
\newcommand{\proofhashonecollisions}{
\shortOnly{\begin{proof}[Proof of \Cref{lem:hashonecollisions}]}\fullOnly{\begin{proof}}
We follow the argument of the proof of \Cref{lem:hashcollisionoblivious} and first analyze the probability of having a large number of iterations with hash collisions or even just a large number of dangerous iterations if the adversary is oblivious. In particular, the probability of having $\Theta(n \eps)$ hash collisions in $\Theta(n \eps)$ dangerous rounds becomes an arbitrarily small $\eps^{\Theta(n \eps)}$ probability if a sufficiently large output length of $o_1 = \Theta(\log \frac{1}{\eps})$ is used in \refalgComputeAdv. 
 Since the number of possible oblivious strategies of selecting at most $2n\eps$ rounds out of at most $2n$ for a corruption is at most $\binom{2n}{2n\eps} < \frac{4}{\eps}^{2n\eps}$ this probability is small enough to take a union bound over all possible adversaries. This extends the proof to fully adversarial channels. 
\end{proof}
}\fullOnly{\proofhashonecollisions}

\Cref{lem:hashonecollisions} implies that even if we treat $hash_1$ hash collisions in \refalgComputeAdv as errors then this only increases the number of possible errors by a constant factor. We can therefore restrict ourselves in the next lemma to analyzing hash collisions due to the $hash_2$ hash function. This hash function however only needs to map the short $o_1 = \Theta(\log \frac{1}{\eps})$ long hash values to a constant output of length $o_2 = \Theta(1)$. Using the hash functions from \Cref{lem:hashes} for this only $\Theta(\log \log \frac{1}{\eps})$ bit sized seeds are necessary. Similar to \Cref{lem:simplehashcollisions} \finalDel{or  \Cref{lem:simplealgconstanthashvalue}} these small seeds are sampled afresh in every iteration which makes the hash collisions due to $hash_2$ being dominated by independent Bernoulli($\Theta(1)$) trials. Again, following the arguments in \Cref{lem:hashcollisionoblivious} having more than $\Theta(n \eps)$ such hash collisions has a probability of at most $2^{-\Theta(n \eps)}$:

\begin{corollary}\label{lem:hashtwocollisions}
For any protocol $\Pi$ and any fully adversarial channel the number of iterations with hash collisions due to hashing with the hash function $hash_2$ in \refalgComputeAdv is at most $\Theta(\eps n)$, with probability $1 - 2^{-\Theta(n \eps)}$.
\end{corollary}
\newcommand{\proofhashtwocollisions}{
\shortOnly{\begin{proof}[Proof of \Cref{lem:hashtwocollisions}]}\fullOnly{\begin{proof}}
We call an iteration in which the hash values of both parties after applying $hash_1$ do not agree a \emph{dangerous} iteration. It is clear that $hash_2$ hash collisions can only occur during dangerous iterations. Let $d$ be the number of dangerous iterations and let $h$ be the number of such $hash_2$ hash collisions (during these iterations). 

According to \Cref{lem:hashonecollisions} the number of iterations with errors or $hash_1$ hash collisions is at most $\Theta(n \eps)$. The exact same calculation as in \Cref{lem:hashcollisionoblivious} therefore shows that $d \leq (\frac{C^{-}}{C^{+}} + 1) h + \Theta(n \eps)$. If $d = \Theta(n \eps)$ is sufficiently large then the fraction of dangerous rounds with a hash collision is at least $\frac{1}{2(\frac{C^{-}}{C^{+}} + 1)}$. If we choose $p_2 = 2^{-o_2} < \frac{1}{4(\frac{C^{-}}{C^{+}} + 1)}$ each dangerous rounds produces a hash collision with a probability which is at least a constant times smaller than this fraction. Since the random $hash_2$ seeds are sampled afresh and independently at the beginning of each iteration the hash collisions are dominated by i.i.d. Bernoulli variables with probability $p_2$ and a Chernoff bound shows that the probability of having this large of a fraction of collisions during dangerous rounds becomes an arbitrarily small $2^{-\Theta(n \eps)}$. 
\end{proof}
}\fullOnly{\proofhashtwocollisions}

With these $\Theta(n \eps)$ bounds on the total number of hash collisions in both \refalgComputeOblivious and \refalgComputeAdv we can prove our main results:

\begin{proof}[Proof of \Cref{lem:mainoblivious} and \Cref{lem:mainadv}]
\Cref{lem:hashcollisionoblivious,lem:hashonecollisions,lem:hashtwocollisions} show that both in \refalgComputeOblivious and \refalgComputeAdv at most $\Theta(n \eps)$ hash collisions or errors happen. \Cref{lem:potentialincreasetwo} shows that the potential drop in these rounds is bounded by a fixed $\Theta(n \eps)$ while the potential increases by at least one in the remaining $R_{total} - \Theta(n \eps)$ rounds. For a sufficiently large $R_{total} = \ceil{n / r} + \Theta(n \eps)$ this implies a potential of at least $\Phi > \ceil{n / r} + \Theta(n \eps)$ in the end. Following the arguments in \Cref{lem:totalpotential} we get that $l^{+} \geq \ceil{n / r}$ which implies that the parties agree upon the first $n$ symbols of the execution of $\Pi$ and therefore both output the correct outcome.

The total round complexity of the main loop in both algorithms is $R_{total} (r + r_c) = (\ceil{n / r} + \Theta(n \eps)) r (1 + \frac{r_c}{r}) = n (1 + \Theta(r \eps)) (1 + \frac{r_c}{r}) = n (1 + \Theta(r \eps + \frac{r_c}{r}))$ and in both algorithms $r$ is set to the (asymptotically) optimal value $r = \ceil{\sqrt{\frac{r_c}{\eps}}}$ which makes this round complexity equal to $n (1 + \Theta(\sqrt{r_c \eps}))$. In \refalgComputeOblivious $r_c = \Theta(1)$ which leads to a round complexity of $n (1 + \Theta(\sqrt{\eps}))$ in the main loop. In \refalgComputeAdv $r_c = \Theta(\log \log \frac{1}{\eps})$ which leads to a round complexity of $n (1 + \Theta(\sqrt{\eps \log \log \frac{1}{\eps}}))$ in the main loop. In both algorithms the communication performed by the randomness exchange is $\Theta(n \sqrt{\eps})$ many rounds and therefore negligible. This shows the communication rate of \refalgComputeOblivious to be $1 - \Theta(\sqrt{\eps})$ and the communication rate of \refalgComputeAdv to be $1 - \Theta(\sqrt{\eps \log \log \frac{1}{\eps}})$ as desired. 
\end{proof}

\finalDel{We conclude with some remarks:

\begin{itemize}
	\item One can also achieve the $1 - \Theta(\sqrt{\eps})$ communication rate against fully adversarial channels in the setting of \cite{FGOS13} where the communicating parties have access to some source of shared randomness that is hidden from the adversary. No such asymptotic gap between having or not having shared randomness exists in the standard one-way setting. To achieve the $1 - \Theta(\sqrt{\eps})$ communication rate one simply uses \refalgComputeOblivious but instead of presharing randomness using the Robust Randomness Exchange, which would expose the randomness to the fully adaptive adversary while it still has the possibility to influence the transcripts to be hashed, parties simply use fresh shared randomness in \Cref{algline:hashingbegin}. 
	\item Similarly, if one assumes a computationally bounded fully adaptive adversary and sufficiently strong computational hardness assumptions to allow for key-agreement the parties can simulate a hidden shared random source by securely exchanging a short random seed at the beginning of the algorithm and then using a cryptographic PRG to stretch this randomness. Since the resulting random string is indistinguishable from a truly random string by the adversary one can use it as a shared hidden source of randomness and again use \refalgComputeOblivious.
	\item A straight forward implementation of our algorithms runs in quadratic time, because each of the $O(n)$ iterations requires a hashing step over an $O(n)$ long transcript. Using ideas similar to \cite{GH13}, but much simpler, one can also achieve a near linear computational complexity. 
	\item One efficient way to implementation of our algorithms, especially in settings where extensive computations are performed between communication steps, is to use \emph{checkpointing}, that is, storing a snapshot of the current application state to enable fast back-tracking of (non-reversible) computations. It is possible to have an implementation in which never more than $\log n$ checkpoints are stored while the total amount of extra computation steps compared to a noise-free execution is at most $\Theta(n \sqrt{\eps})$. 
\end{itemize}
}

\bibliographystyle{abbrv}
\bibliography{ref}

\finalDel{

\newpage

\appendix

\setlength{\parskip}{0.2cm}

\shortOnly{
\section{Important Remarks Regarding the Interactive Coding Settings}\label{sec:importantremarks}

\importantremarks
}

\section{Discussion of the $1 - \Theta(\sqrt{\eps \log \frac{1}{\eps}})$ Bound of \cite{KR}}\label{app:hepsboundexplanaition}

In this section we try to briefly give some explanation and intuition regarding the extra assumptions on the coding scheme the were made in \cite{KR} and how these assumptions lead to the slightly stronger upper bound of $1 - \Omega(\sqrt{\eps \log \frac{1}{\eps}})$ which, as the algorithms in this paper show, disappears for alternating protocols or for adaptive-simulations. In particular, the following non-adaptivity assumptions of \cite{KR} is crucial in proving this bound: 

The order in which Alice and Bob talk during the simulation is predetermined a priori and therefore independent of when errors happen. In particular, this explicitly forbids Bob to adapt the length of a clarification provided to Alice depending on whether or how many errors have happened or how much Alice (reportedly) already understood. 

The protocol which was chosen by the impossibility result of \cite{KR} to be simulated is furthermore structurally more complex than an alternating protocol: \cite{KR} essentially assumes a uniformly random protocol over an alphabet of much larger bit size. In particular, the bit size $B$ grows with $1/\eps$ and is chosen to be $B = \frac{\sqrt{\log \frac{1}{\eps}}}{\sqrt{\eps}}$. This means that communicating any of these messages on its own requires more than the $r = \frac{1}{\sqrt{\eps}}$ rounds we have until we need to add redundancy.

In such an assumed setting Alice could, as before, try to add a parity check bit to detect errors, for example, at the end of her first question. Because of the non-adaptivity assumption however, even when an error was detected the predetermined order does not allow the parties to adjust their communication adaptively. In particular, the parties need to essentially decide a priori how much time Alice spends to communicate the first $B$ bit long question before Bob starts answering. If, by this pre-allocated time, Bob is not clear on the question his whole slot, which is reserved for his $B$ bit long answer, will essentially go to waste. A single error happens with probability $\eps B$ and if one fails to resolve it with constant probability this would lead to a rate loss of $\Theta(\frac{\eps B^2}{B}) = \Theta(\eps B)$. Since the value of $B$ is chosen large enough by the second assumption this would be a rate loss which cannot be tolerated. Alice therefore needs to determine a priori how many redundant steps she needs to add to allow Bob to resolve one error (the unlikely case of more than one error which happens with probability $\Theta((\eps B)^2)$ can be ignored). This however means that Alice needs to send the position of the erroneous symbol. This requires $\log B = \log \frac{1}{\eps}$ bits. In short, over a binary channel detecting an error costs only one (parity) bit while correcting one error requires $\log B$ extra bits. Overall these $\Theta(\log B)$ extra transmissions for an error correction need to be planned in for every $B$ bit answer or question even if both parties are aware that no error happened. This leads to a rate loss of $\log B / B$. The overall rate loss is therefore either $\eps B$ or $\log B / \eps$ which for $B = \frac{\sqrt{\log \frac{1}{\eps}}}{\sqrt{\eps}}$ is a rate loss of $1 - \Theta(\sqrt{\frac{1}{\eps} \log \frac{1}{\eps}})$ either way. 

In summary, having a large $B$ in combination with making the algorithm to decide a priori who talks at what time forces the algorithm to supply not just sufficient information to detect errors (since adaptive backtracking is not possible) but it needs to plan in enough redundancy ahead of time to be able to also correct these errors. The number of symbols needed for such a correction is $\log B$ for a binary channel. This forces the algorithm to waste $\log B$ bits of error correction for every answer even if the parties have already determined that no error has happened. The combination of a predetermined order and a sufficiently non-regular protocol to be simulated therefore leads to an overall rate loss of $1 - \Theta(\sqrt{\frac{1}{\eps} \log \frac{1}{\eps}})$. On the other hand allowing for adaptive algorithms or simulating any periodic protocol with period at most $\frac{1}{\sqrt{\eps}}$ allows for a better communication rate of $1 - \Theta(\sqrt{\eps})$ at which point one hits a hard and fundamental limit as shown in \Cref{sec:impossibility}.

\shortOnly{\reducingseedlength}

\shortOnly{
\section{Proofs}\label{app:proofs}

\proofseedlengthLB

\bigskip
\proofpotentialincreasetwo

\bigskip
\prooftotalpotential

\bigskip
\proofhashcollisionoblivious

\bigskip
\proofhashonecollisions

\bigskip
\proofhashtwocollisions
}

}

\end{document}